\documentclass{jams-l-m}

\usepackage{amssymb}

\usepackage{dsfont}

\usepackage{mathabx}

\usepackage{graphicx}

\usepackage{hyperref}

\usepackage{xcolor}

\usepackage{mathbbol}

\usepackage{array} 

\usepackage{float}

\usepackage{soul} 

\usepackage[autostyle]{csquotes} 

\newtheorem{theorem}{Theorem}[section]

\newtheorem{proposition}[theorem]{Proposition}

\theoremstyle{definition}
\newtheorem{definition}[theorem]{Definition}
\newtheorem{example}[theorem]{Example}

\theoremstyle{remark}
\newtheorem{remark}[theorem]{Remark}

\numberwithin{equation}{section}

\DeclareSymbolFontAlphabet{\amsmathbb}{AMSb}%

\usepackage{tikz}

\usetikzlibrary{arrows,positioning,calc} 
\usetikzlibrary{decorations.pathreplacing}
\tikzset{
    >=stealth',
    punkt/.style={
           rectangle,
           rounded corners,
           draw=black, thick,
           text width=5.5em,
           minimum height=2em,
           text centered},
    punktl/.style={
           rectangle,
           rounded corners,
           draw=black, thick,
           text width=7em,
           minimum height=2em,
           text centered},
    pil/.style={
           ->,
           shorten <=4pt,
       shorten >=4pt
    },
    pildotted/.style={
           ->,
           shorten <=4pt,
           shorten >=4pt,
  dotted,
  },
    punktf/.style={
           rectangle,
           text width=4.0em,
           minimum height=1.5em,
           text centered},
    punktfTop/.style={
           rectangle,
           text width=4.0em,
           minimum height=1.5em,
           text centered,
           append after command={
               [thick,shorten >=0.2bp, shorten <=0.2bp]
               (\tikzlastnode.north west)edge(\tikzlastnode.north east)
}
    },
    punktfBot/.style={
           rectangle,
           text width=4.0em,
           minimum height=1.5em,
           text centered,
           append after command={
               [thick,shorten >=0.2bp, shorten <=0.2bp]
               (\tikzlastnode.south west)edge(\tikzlastnode.south east)
            }
    }
}

\newcommand{\ti}[1]{\widetilde{#1}}

\usepackage{prodint}
\theoremstyle{plain}

\newtheorem{condition}{Condition}

\makeatletter
\newcommand*\bigcdot{\mathpalette\bigcdot@{.65}}
\newcommand*\bigcdot@[2]{\mathbin{\vcenter{\hbox{\scalebox{#2}{$\m@th#1\bullet$}}}}}
\makeatother

\usepackage[shortlabels]{enumitem} 

\usepackage{changepage}

\begin{document}

\title[Disability insurance with collective health claims]{Disability insurance with collective health claims: a mean-field approach}


\author{Christian Furrer}
\address{Department of Mathematical Sciences, University of Copenhagen, Universitetsparken 5, DK-2100 Copenhagen, Denmark}
\curraddr{}
\email{\href{mailto:furrer@math.ku.dk}{furrer@math.ku.dk}}
\thanks{}

\author{Philipp C.\ Hornung}
\address{Department of Mathematical Sciences, University of Copenhagen, Universitetsparken 5, DK-2100 Copenhagen, Denmark}
\curraddr{}
\email{\href{mailto:pcho@math.ku.dk}{pcho@math.ku.dk}}
\thanks{}


\date{}

\dedicatory{\today}

\begin{abstract}

The classic semi-Markov disability model is expanded with individual and collective health claims to improve its explanatory and predictive power -- in particular in the context of group experience rating.  The inclusion of collective health claims leads to a computationally challenging many-body problem. By adopting a mean-field approach, this many-body problem can be approximated by a non-linear one-body problem, which in turn leads to a transparent pricing method for disability coverages based on a lower-dimensional system of non-linear forward integro-differential equations. In a practice-oriented simulation study, the mean-field approximation clearly stands its ground in comparison to naïve Monte Carlo methods. \\

\noindent\textbf{Keywords:} Group experience rating; non-linear forward equations; semi-Markov model.

\end{abstract}

\maketitle



\section{Introduction}

Disability insurance plays a vital role in ensuring income stability and supporting part-time employment during periods with reduced earning capacity. In many countries, disability coverages are sold not only directly to individuals, but also as part of a company pension scheme. It is therefore essential for insurers to be able to accurately price such coverages -- taking into account the fact that the physical and psychological work environment of each company likely has a substantial effect on the frequency, but also the severity, of disability claims. This suggests the application of group experience rating to disability insurance, which has been explored via an empirical Bayes approach for Markov chains in~\cite{Furrer2019,FurrerSoerensenYslas2025}. However, since the frequency of disability claims is rather low and single claims with long durations tend to have a substantial impact on the total loss, such a direct approach to experience rating is greatly challenged. Simply put: It can be near impossible to distinguish between `good' and `bad' companies, since there is just too limited data available.

In this paper, we address the `small data' issue {of infrequent disability claims} by drawing on three observations. First, disability insurance and health insurance are increasingly, at least in certain countries such as Denmark, sold together. Second, it is reasonable to expect -- at least when adjusting for covariates -- that there is a relationship between an employees disability frequency and the extent of health claims across all employees; this observation is indirectly based on the assumption that both factors are largely attributable to the physical and psychological working environment.  For example, a toxic work culture would lead to increased mental health risk, which would show itself in the form of health claims (consulting a psychologist, etc.) and, later, in actual disability (loss of working capacity due to severe stress). Third and finally, the scope of health insurance data is much more extensive, as this type of insurance is often used on a regular basis. Based on these observations, we formulate a multi-state model for disability insurance with integrated information about health claims.

To be specific, we expand the classic semi-Markov disability model (see~\cite{Hoem1972,Helwich2008,Christiansen2012,BuchardtMollerSchmidt2015}) with {collective and individual health claims}. In the classic model, transition probabilities, expected cash flows, and prospective reserves may efficiently be calculated using Kolmogorov's integro-differential equations. However, if we denote the number of individuals by $n$, the inclusion of collective health claims entails that the computational complexity of the relevant forward equations grows exponentially in $n$. Consequently, we adopt the mean-field approach outlined in~\cite{Hornung2025} to obtain, as an approximation in the limit $n \to \infty$, a lower-dimensional system of non-linear forward integro-differential equations. Furthermore, we briefly address statistical aspects as well as practical implementation.

In general, the insurance liabilities of an individual may depend on the other individuals either through the payments or, as is the case here, because the individuals themselves are dependent. In the area of actuarial multi-state modeling, mean-field approximations have hitherto received the most interest in the former case, see in particular~\cite{DjehicheLoefdahl2016}. In light hereof, we believe this paper offers a fresh perspective on the application of mean-field theory to the actuarial field.

The remainder of the paper is organized as follows. In Section~\ref{sec:setup}, we expand the classic semi-Markov disability model with individual and collective health claims. Section~\ref{sec:mean_field} contains a description of the corresponding mean-field model, including the associated system of non-linear forward integro-differential equations, and proofs of the required convergences. The next two sections are more oriented towards practice. In Section~\ref{sec:practical}, we introduce and briefly discuss a meta-algorithm for solving the relevant system of differential equations, while Section~\ref{sec:sim_study} is devoted to a simulation study and the comparison of the mean-field approximation to naïve Monte Carlo methods. The paper concludes with Section~\ref{sec:conclude} in which we offer some extensions and avenues for future work.

\section{Disability model with health claims} \label{sec:setup}

This section is devoted to the expansion of the classic semi-Markov disability model with initially individual health claims and subsequently also collective health claims.

\subsection{Semi-Markov model}\label{sec:pre_mean-field}

Let $(\Omega,\mathcal{F},\amsmathbb{F},\amsmathbb{P})$ be a filtered probability spac{e} and write $\amsmathbb{F}=(\mathcal{F}_t)_{t\geq0}$. The multi-state approach to classic disability insurance considers a non-explosive jump process $Z$ on the finite state space $\mathcal{J} = \{1,2,3\}$ according to Figure~\ref{fig:disability_state_space}. The initial distribution of $Z$ is denoted $\pi$. We associate to $Z$ a multivariate counting process $N$ with components $N_{jk}$, $j \neq k$, given by
\begin{align*}
N_{jk}(t) = \#\{s \in (0,t] : Z_{s-} = j, Z_s = k\}.
\end{align*}
In this paper, we consider contractual payments prescribed by a payment process $B$ given by
\begin{align}\label{eq:sec2:B}
B(\mathrm{d}t)
=
\sum_j \mathds{1}_{\{Z_{t-} = j\}} b_j(t,U_{t-}) \mathrm{d}t + \sum_{j \neq k} b_{jk}(t,U_{t-}) N_{jk}(\mathrm{d}t), 
\end{align}
where $U$ is the duration process associated with $Z$ given by
\begin{align*}
U_t = t - \sup\{s \in [0,t] : Z_s \neq Z_t\} \text{ for } t > 0 \text{ and } U_0 = 0,
\end{align*}
while $(t,u) \mapsto b_j(t,u)$ and $(t,u) \mapsto b_{jk}(t,u)$ are measurable functions that are bounded on compacts and which describe sojourn payment rates and transition payments, respectively. 

\begin{figure}[ht!]
    \centering
    \scalebox{0.8}{
    \begin{tikzpicture}[node distance=8em, auto]
	\node[punkt] (i1) {Active};
        \node[anchor=north east, at=(i1.north east)]{$1$};
        \node[punkt, right=3cm of i1] (i2) {Disabled};
        \node[anchor=north east, at=(i2.north east)]{$2$};
        \node[, right=1.5cm of i1] (a) {};
        \node[punkt, below=1.5cm of a] (a2) {Dead};
        \node[anchor=north east, at=(a2.north east)]{$3$};
        \path (i1) edge [pil, bend right=15] 	node [below=0.15cm]	{} (i2);
        \path (i2) edge [pil, bend right=15] 	node [above=0.1cm] 	{} (i1);
        \path (i1) edge [pil] 						node [left=0.15cm]  	{} (a2);
        \path (i2) edge [pil] 						node [right=0.15cm] 	{} (a2);
    \end{tikzpicture}}
    \caption{State space $\mathcal{J}=\{1,2,3\}$ for classic disability insurance. The arrows represent the possible transitions.}
    \label{fig:disability_state_space}
\end{figure}
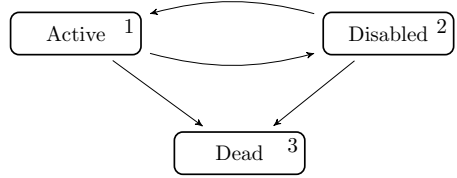

Semi-Markov modeling entails the assumption that $(Z,U)$ is Markov, and smooth semi-Markov modeling adds the assumption that the {distributions of the} jump times should be absolutely continuous (with respect to the Lebesgue measure). An alternative, but equivalent, formulation is in terms of the compensators of the multivariate counting process, which then should read
\begin{align*}
\amsmathbb{E}[N_{jk}(\mathrm{d}t) \, | \, \mathcal{F}_{t-}]  = \mathds{1}_{\{Z_{t-} = j\}} \mu_{jk}(t,U_{t-}) \mathrm{d}t.
\end{align*}
The functions $(t,u) \mapsto \mu_{jk}(t,u)$, $j \neq k$, are the so-called (duration-dependent) transition rates, which we assume to be measurable and bounded on compacts.

{Throughout, we assume implicitly that there exists a maximal contract time $T \in (0,\infty)$ after which all payments are zero. We may then define the state-wise prospective reserves at contract inception, $(V_i)_i$, according to}
\begin{align*}
V_i
=
\amsmathbb{E}\bigg[\int_0^T e^{-\int_0^t r(s) \, \mathrm{d}s} B(\mathrm{d}t) \, \bigg| \, Z_0 = i\bigg],
\end{align*}
where the deterministic interest rate $t \mapsto r(t)$ is a measurable function that is bounded on compacts. If we denote by $(A_i)_i$ the state-wise expected accumulated cash flows given by
\begin{align*}
A_i(\mathrm{d}t) = \amsmathbb{E}[B(\mathrm{d}t) \, | \, Z_0 = i],
\end{align*}
then it holds that
\begin{align*}
V_i = \int_0^T e^{-\int_0^t r(s) \, \mathrm{d}s} A_i(\mathrm{d}t).
\end{align*}
In similar fashion, the portfolio-wide prospective reserve at contract inception $V$ reads
\begin{align*}
V
=
\amsmathbb{E}\bigg[\int_0^T e^{-\int_0^t r(s) \, \mathrm{d}s} B(\mathrm{d}t) \bigg]
=
\int_0^T e^{-\int_0^t r(s) \, \mathrm{d}s} A(\mathrm{d}t),
\end{align*}
where the portfolio-wide expected accumulated cash flow $A$ is given by
\begin{align*}
A(\mathrm{d}t) = \amsmathbb{E}[B(\mathrm{d}t)].
\end{align*}
In the following, let
\begin{align*}
p_{ij}(t,u) := \amsmathbb{P}(Z_t = j, U_t \leq u \, | \, Z_0 = i) \text{ and } p_j(t,u) := \amsmathbb{P}(Z_t = j, U_t \leq u).
\end{align*}
be the transition probabilities and occupation probabilities, respectively. Note that
\begin{align*}
p_j = \sum_i \pi(i) p_{ij}, \quad A = \sum_i \pi(i) A_i, \quad V = \sum_i \pi(i) V_i.
\end{align*}

\begin{example}
Consider again a disability annuity with a waiting period of $\varepsilon\geq0$, corresponding to $b_2(t,u) = b \mathds{1}_{\{u \geq \varepsilon\}}$ with $b>0$, while all other payments are zero.
\begin{align*}
A_i(\mathrm{d}t) 
&=
\int_0^t b\mathds{1}_{\{u \geq \varepsilon\}}p_{i2}(t,\mathrm{d}u)\,\mathrm{d}t
=
b\mathds{1}_{\{t \geq \varepsilon\}} \big(p_{i2}(t,t) - p_{i2}(t,\varepsilon)\big)\mathrm{d}t, \\
A(\mathrm{d}t) 
&=
\int_0^t b\mathds{1}_{\{u \geq \varepsilon\}}p_2(t,\mathrm{d}u)\,\mathrm{d}t
=
b\mathds{1}_{\{t \geq \varepsilon\}} \big(p_2(t,t) - p_2(t,\varepsilon)\big)\mathrm{d}t.
\end{align*}
Consequently, 
\begin{align*}
V_i &= b \int_\varepsilon^T e^{-\int_0^t r(s) \, \mathrm{d}s} \big(p_{i2}(t,t) - p_{i2}(t,\varepsilon)\big)\mathrm{d}t, \\
V &= b \int_\varepsilon^T e^{-\int_0^t r(s) \, \mathrm{d}s} \big(p_2(t,t) - p_2(t,\varepsilon)\big)\mathrm{d}t.
\end{align*}
\end{example}
The following results are standard -- see for instance~\cite{BuchardtMollerSchmidt2015}.
\begin{proposition}\label{prop:semi_markov_cf}
It holds that
\begin{align*}
A_i(\mathrm{d}t) &= \sum_j \int_0^t \Big(b_j(t,u)+ \sum_{k : k \neq j} b_{jk}(t,u) \mu_{jk}(t,u) \Big)  p_{ij}(t,\mathrm{d}u) \, \mathrm{d}t, \\
A(\mathrm{d}t) &= \sum_j \int_0^t \Big(b_j(t,u) + \sum_{k : k \neq j} b_{jk}(t,u) \mu_{jk}(t,u)\Big)  p_j(t,\mathrm{d}u) \, \mathrm{d}t.
\end{align*}
\end{proposition}
\begin{proposition}\label{prop:semi_markov_forward}
Let $d\geq0$. It holds almost everywhere on $[d,\infty)$ that
\begin{align*}
\frac{\mathrm{d}}{\mathrm{d}t} p_{ij}(t,t-d)
=
\sum_{k : k \neq j} \int_0^t \mu_{kj}(t,u) p_{ik}(t,\mathrm{d}u)
-
\int_0^{t-d} \mu_{j\bigcdot}(t,u) p_{ij}(t,\mathrm{d}u)
\end{align*}
with boundary conditions $p_{ij}(t,0) = \mathds{1}_{\{t = 0\}}\mathds{1}_{\{i = j\}}$. It further holds almost everywhere on $[d,\infty)$ that
\begin{align*}
\frac{\mathrm{d}}{\mathrm{d}t} p_j(t,t-d)
=
\sum_{k : k \neq j} \int_0^t \mu_{kj}(t,u) p_k(t,\mathrm{d}u)
-
\int_0^{t-d} \mu_{j\bigcdot}(t,u) p_j(t,\mathrm{d}u)
\end{align*}
with boundary conditions $p_j(t,0) = \mathds{1}_{\{t=0\}}\pi(j)$.
\end{proposition}
\begin{remark}
In light of the reparameterization of Remark~\ref{rmk:semi_markov_reparam} below, Proposition~\ref{prop:semi_markov_forward} states that the transition probabilities solve a version of Kolmogorov's forward equations. In fact, they \textit{uniquely} solve these equations, confer with~\cite{FeinbergMandavaShiryaev2014,FeinbergMandavaShiryaev2022}.
\end{remark}
\begin{remark}\label{rmk:occ_to_trans}
Note that the occupation probabilities may be calculated directly or by first calculating the transition probabilities and then taking a weighted average with respect to the initial distribution; the latter would of course be more time-consuming. In particular, the occupation probability $p_j$ corresponds to the transition probability $p_{ij}$ if $\pi(i) = 1$. This is due to the fact that they have the same evolution forward through time and this evolution does not depend on the initial distribution. We emphasize this at the present time because precisely this property is lost as we turn to mean-field approximations in Section~\ref{sec:mean_field}.
\end{remark}

The terms of the integro-differential equations for the transition probabilities admit an intuitive interpretation. The first term is the aggregate probability mass stemming from paths starting in $i$ at time zero, transitioning into state $k$ at some time $t-u\in[0,t)$, and staying in state $k$ until transitioning into state $j$ at exactly time $t$ (inflow). Similarly, the second term subtracts the probability mass stemming from paths starting in state $i$ at time zero, transitioning into $j$ no earlier than time $d$ (as the duration at time $t$ should be smaller than $t-d$), and staying there until transitioning out of state $j$ at exactly time $t$ (outflow).

The systems of integro-differential equations of Proposition~\ref{prop:semi_markov_forward} can be solved using, for instance, the algorithm presented in Section~3 of~\cite{BuchardtMollerSchmidt2015}. Numerical schemes and practical implementation is discussed in greater detail in the latter Section~\ref{sec:practical}.

\begin{remark}\label{rmk:semi_markov_reparam}
The semi-Markov model can be reparameterized, using the process $Y$ given by $Y_t = t-U_t$, $t\geq0$, instead of $U$. Since $U$ is the duration \textit{since} the last jump, $Y$ is the time \textit{of} the last jump. Therefore, while $(Z,U)$ is a piecewise deterministic process, $(Z,Y)$ is actually a jump process. This proves mathematically convenient.

In fact, the model $(Z,U)$ with transition rates $\mu_{jk}$, $j \neq k$, and transition probabilities $p_{ij}$ is equivalent to the model $(Z,Y)$ with transition rates $(t,y) \mapsto \ti{\mu}_{jk}(t,y):=\mu_{jk}(t,t-y)$, $j \neq k$, and transition probabilities $\tilde{p}_{ij}$ given by
\begin{align*}
\ti{p}_{ij}(t,d) =\amsmathbb{P}(Z_t = j, Y_t \geq d \, | \, Z_0 = i ,Y_0=0), \quad 0 \leq d \leq t.
\end{align*}
To see this, we may use the definition of $Y$ to obtain
\begin{align*}
\ti{p}_{ij}(t,d)
&=
\amsmathbb{P}(Z_t = j, Y_t \geq d \, | \, Z_0 = i, Y_0=0) \\
&=
\amsmathbb{P}(Z_t = j, U_t \leq t-d \, | \, Z_0 = i, U_0=0) \\
&=
p_{ij}(t,t-d)
\end{align*}
for $0 \leq d \leq t$. Applying a change of variables, we arrive at a version of Kolmogorov's forward equations for $(Z,Y)$:
\begin{align*}
\frac{\mathrm{d}}{\mathrm{d}t} \ti{p}_{ij}(t,d)
&=
\sum_{k : k \neq j} \int_0^t \ti{\mu}_{kj}(t,y) \ti{p}_{ik}(t,\mathrm{d}y)
-
\int_d^{t} \ti{\mu}_{j\bigcdot}(t,y) \ti{p}_{ij}(t,\mathrm{d}y), \quad 0 \leq d \leq t, \\
\ti{p}_{ij}(0,0) &= \mathds{1}_{\{i = j\}}.
\end{align*}
\end{remark}

\subsection{Extension: Individual health claims}

It is reasonable to assume that the past number of health insurance claims of an individual is informative about the likelihood of a disability claim. Frequent use of the health insurance policy might predate a disability claim and, similarly, limited or no use might indicate good overall health. While disability claims are rare, health claims are frequent; this makes it particularly attractive to utilize the latter in risk profiling the individual. Therefore, we add to the model a {non-explosive} counting process $H$ describing the number of health insurance claims of the individual. Foremost for notational convenience, we make the simplifying assumption that $H$ and $N$ admit no simultaneous jumps. Further, we assume that 
\begin{align*}
\amsmathbb{E}[N_{jk}(\mathrm{d}t) \, | \, \mathcal{F}_{t-}]
&=
\mathds{1}_{\{Z_{t-} = j\}} \mu_{jk}(t,U_{t-},H_{t-}) \mathrm{d}t, \\
\amsmathbb{E}[H(\mathrm{d}t) \, | \, \mathcal{F}_{t-}] 
&=
\lambda_{Z_{t-}}(t,U_{t-},H_{t-}) \mathrm{d}t.
\end{align*}
The functions $(t,u,h) \mapsto \mu_{jk}(t,u,h)$, $j \neq k$, are duration- and health-dependent transition rates, which we assume to be measurable and bounded on compacts. The functions $(t,u,h) \mapsto \lambda_j(t,u,h)$ are health claim hazards of similar nature, which we also assume to be measurable and bounded on compacts. It follows that $(Z,U,H)$ is a Markov process{, confer for instance with Theorem~7.3.1(a) in~\cite{Jacobsen2006}.}

\begin{example}\label{ex:health_claims}
In modeling the health claims, a simple choice would be
\begin{align*}
\amsmathbb{E}[H(\mathrm{d}t) \, | \, \mathcal{F}_{t-}]  = \lambda_{Z_{t-}} \, \mathrm{d}t,
\end{align*}
for non-negative constants $\lambda_j$. In this case, $H$ is a doubly stochastic Poisson process with hazards depending only on the state (active, disabled, and dead) of the individual.
\end{example}

In the following, let
\begin{align*}
p_{ij}(t,u,h) :=&\, \amsmathbb{P}(Z_t = j, U_t \leq u, H_t=h \, | \, Z_0 = i), \\
p_j(t,u,h) :=&\, \amsmathbb{P}(Z_t = j, U_t \leq u,H_t=h)
\end{align*}
be the transition probabilities and occupation probabilities, respectively. Note that $p_j = \sum_i \pi(i) p_{ij}$.

\begin{example}\label{ex:waiting_period_continued}
Consider again a disability annuity with a waiting period of $\varepsilon\geq0$, corresponding to $b_2(t,u) = b \mathds{1}_{\{u \geq \varepsilon\}}$ with $b>0$, while all other payments are zero. In that case, we now have that 
\begin{align*}
A_i(\mathrm{d}t) 
&=
\sum_{h=0}^\infty \int_0^t b\mathds{1}_{\{u \geq \varepsilon\}}p_{i2}(t,\mathrm{d}u,h)\,\mathrm{d}t
=
b\mathds{1}_{\{t \geq \varepsilon\}} \sum_{h=0}^\infty\big(p_{i2}(t,t,h) - p_{i2}(t,\varepsilon,h)\big)\mathrm{d}t, \\
A(\mathrm{d}t) 
&=
\sum_{h=0}^\infty \int_0^t b\mathds{1}_{\{u \geq \varepsilon\}}p_2(t,\mathrm{d}u,h)\,\mathrm{d}t
=
b\mathds{1}_{\{t \geq \varepsilon\}} \sum_{h=0}^\infty\big(p_2(t,t,h) - p_2(t,\varepsilon,h)\big)\mathrm{d}t.
\end{align*}
Consequently,
\begin{align*}
V_i &= b\int_\varepsilon^T e^{-\int_0^t r(s) \, \mathrm{d}s} \sum_{h=0}^\infty\big(p_{i2}(t,t,h) - p_{i2}(t,\varepsilon,h)\big)\mathrm{d}t, \\
V &= b\int_\varepsilon^T e^{-\int_0^t r(s) \, \mathrm{d}s}\sum_{h=0}^\infty \big(p_2(t,t,h) - p_2(t,\varepsilon,h)\big)\mathrm{d}t.
\end{align*}
Since the payments themselves do not depend on $H$, this variable is simply summated out. 
\end{example}

The following two propositions expand Propositions~\ref{prop:semi_markov_cf} and~\ref{prop:semi_markov_forward} to also include individual health claims.

\begin{proposition}\label{prop:semi_markov_health_cf}
It holds that
\begin{align*}
A_i(\mathrm{d}t)
&=
\sum_j\sum_{h=0}^{\infty}\int_0^t \Big(b_j(t,u)+ \sum_{k : k \neq j} b_{jk}(t,u) \mu_{jk}(t,u,h) \Big)  p_{ij}(t,\mathrm{d}u,h) \, \mathrm{d}t, \\
A(\mathrm{d}t)
&=
\sum_j \sum_{h=0}^{\infty} \int_0^t \Big(b_j(t,u) + \sum_{k : k \neq j} b_{jk}(t,u) \mu_{jk}(t,u,h)\Big)  p_j(t,\mathrm{d}u,h) \, \mathrm{d}t.
\end{align*}
\end{proposition}
\begin{proof}
Referring to the martingale property of differences between the counting processes and their (absolutely continuous) compensators, we find that
\begin{align*}
A_i(\mathrm{d}t)
=
\amsmathbb{E}\Big[ b_{Z_t}(t,U_t) + \sum_{k : k \neq Z_t} b_{Z_t k}(t,U_t)  \mu_{Z_t k}(t,U_t,H_t) \, \Big| \, Z_0 = i\Big] \mathrm{d}t. 
\end{align*}
Invoking the relevant push-forward measure, the first statement of the proposition is immediate. The second statement is proveable in an identical fashion.
\end{proof}

In the following, we adopt the convention that $(t,u) \mapsto \lambda_j(t,u,-1)$ is constantly zero.

\begin{proposition}\label{prop:semi_markov_forward_health}
 Let $d\geq0$. It holds almost everywhere on $[d,\infty)$ that
\begin{align*}
&\frac{\mathrm{d}}{\mathrm{d}t} p_{ij}(t,t-d,h) \\
&=
\sum_{k : k \neq j} \int_0^t \mu_{kj}(t,u,h) p_{ik}(t,\mathrm{d}u,h)-\int_0^{t-d} \mu_{j\bigcdot}(t,u,h) p_{ij}(t,\mathrm{d}u,h) \\
&\quad+
\int_{0}^{t-d}\lambda_j(t,u,h-1)p_{ij}(t,\mathrm{d}u,h-1)-\int_0^{t-d}\lambda_j(t,u,h) p_{ij}(t,\mathrm{d}u,h) 
\end{align*}
with boundary conditions $p_{ij}(t,0,h) = \mathds{1}_{\{t=0\}}\mathds{1}_{\{h=0\}}\mathds{1}_{\{i = j\}}$. It further holds almost everywhere on $[d,\infty)$ that
\begin{align*}
&\frac{\mathrm{d}}{\mathrm{d}t} p_{j}(t,t-d,h) \\
&=
\sum_{k : k \neq j} \int_0^t \mu_{kj}(t,u,h) p_{k}(t,\mathrm{d}u,h)-\int_0^{t-d} \mu_{j\bigcdot}(t,u,h) p_{j}(t,\mathrm{d}u,h)\\
&\quad+
\int_{0}^{t-d}\lambda_j(t,u,h-1)p_{j}(t,\mathrm{d}u,h-1)-\int_0^{t-d}\lambda_j(t,u,h) p_{j}(t,\mathrm{d}u,h) 
\end{align*}
with boundary conditions $p_{j}(t,0,h) = \mathds{1}_{\{t=0\}}\mathds{1}_{\{h=0\}}\pi(j)$.
\end{proposition}
\begin{proof}
The second statement of the proposition follows, for instance, from the first statement and the identity $p_j = \sum_i \pi(i) p_{ij}$.  To prove the first statement, we consider the trivariate process $(Z,Y,H)$ with $Y$ defined according to $Y_t = t-U_t$, $t\geq0$; this is a a Markov jump process. By definition, the compensators of the counting processes associated with $(Z,Y,H)$ are predictable w.r.t.\ the information generated by themselves, and thus by Theorem 4.8.1 of~\cite{Jacobsen2006}, the distribution of $(Z,Y,H)$ is fully characterized by these compensators. Consequently, we may invoke~(7.17) on p.~151 in~\cite{Jacobsen2006} to obtain a system of integro-differential equations which, after a change of variable similar to Remark~\ref{rmk:semi_markov_reparam}, yields the desired result.
\end{proof}

\begin{remark}
In light of the reparameterization utilized in the proof, Proposition~\ref{prop:semi_markov_forward_health} states that the transition probabilities solve a version of Kolmogorov's forward equations. In fact, they \textit{uniquely} solve these equations, confer with~\cite{FeinbergMandavaShiryaev2014,FeinbergMandavaShiryaev2022}.
\end{remark}

\begin{remark}
Note that the occupation probabilities may be calculated directly or by first calculating the transition probabilities and then taking a weighted average with respect to the initial distribution. In particular, the occupation probability $p_j$ corresponds to the transition probability $p_{ij}$ if $\pi(i)=1$. This is due to the fact that they have the same evolution forward through time and this evolution does not depend on the initial distribution. We emphasize this once again because precisely this property is lost as we turn to mean-field approximations in Section~\ref{sec:mean_field}.
\end{remark}

The addition of individual health claims has not hurt the intuitive interpretation of the terms of the integro-differential equations. The first line is concerned with jumps of $Z$, while the second line is concerned with jumps of $H$; this split is possible since simultaneous jumps were disallowed -- otherwise, two additional terms would be necessary. The two terms of the first line are similar to before, but with an added requirement of $H$ having reached state $h$ strictly before time $t$. The two terms of the second line are new. The first term adds the aggregate probability mass stemming from paths for which the $h$'th health claim occurs at time $t$ (inflow), while the second term subtracts the aggregate probability mass stemming from path for which the $(h+1)$'th health claim occurs at time $t$ (outflow).

The systems of integro-differential equations of Proposition~\ref{prop:semi_markov_forward} can, for instance, be solved by adapting the algorithm presented in Section~3 of ~\cite{BuchardtMollerSchmidt2015}. Compared to the situation without individual health claims, the procedure would have to be used for every triplet $(i,j,h)$ rather than only every pair $(i,j)$. In practice, one must choose a cut-off $K_H$, calculate the transition probabilities for $h=0,\ldots,K_H$, and perhaps extrapolate for $h>K_H$. Numerical schemes and practical implementation is discussed in greater detail in the latter Section~\ref{sec:practical}.

\subsection{Extension: Collective health claims}\label{subsec:collective}
So far, we have only considered a single individual -- implicitly assuming individuals to be independent. However, one could easily imagine this to not be the case. There is of course the obvious example of events and trends that impact society as a whole, such as pandemics and climate change, but also technological and medical advancements. However, as explained in the introduction, we rather have the example of company level insurance plans and collective health claims in mind. Dependence between individuals here stems from the fact that each employee is affected by the same physical and psychological work environment. If the insurance plan contains both disability and health coverage, aggregate information about health claims, which are much more extensive than disability claims, could serve as a reasonable predictor of, for example, the disability rate.

To formalize a model able to capture the stylized facts, we consider $n\in\amsmathbb{N}$ individuals, each with an associated trivariate process $X^{\ell,n} = (Z^{\ell,n}, U^{\ell,n}, H^{\ell,n})$. Foremost for notational convenience, we make the additional simplifying assumption that $H^{1,n},N^{1,n},\ldots,H^{n,n},N^{n,n}$ admit no simultaneous jumps. We denote by $E$ the state space of each $X^{\ell,n}$, that is $E = \mathcal{J} \times [0,\infty) \times \amsmathbb{N}_0$, and consider a measurable {and bounded function $g : E \mapsto \amsmathbb{R}^d$, $d \in \amsmathbb{N}$.} Based hereon, we define the averaged process $\nu^n$ according to
\begin{align*}
\nu_t^n = \frac{1}{n} \sum_{\ell = 1}^n g(X_t^{\ell,n}),
\end{align*}
and we assume that
\begin{equation}\label{eq:interacting_compensators}
\begin{aligned}
\amsmathbb{E}[N_{jk}^{\ell,n}(\mathrm{d}t) \, | \, \mathcal{F}_{t-}]
&=
\mathds{1}_{\{Z_{t-}^{\ell,n} = j\}} \mu_{jk}(t,U_{t-}^{\ell,n},H_{t-}^{\ell,n},\nu_{t-}^n) \mathrm{d}t, \\
\amsmathbb{E}[H^{\ell,n}(\mathrm{d}t) \, | \, \mathcal{F}_{t-}] 
&=
\lambda_{Z_{t-}^{\ell,n}}(t,U_{t-}^{\ell,n},H_{t-}^{\ell,n},\nu_{t-}^n) \mathrm{d}t.
\end{aligned}
\end{equation}
The functions $(t,u,h,y) \mapsto \mu_{jk}(t,u,h,y)$, $j \neq k$, are duration-, health- as well as collective-dependent transition rates, which we assume to be measurable and bounded on compacts. The functions $(t,u,h,y) \mapsto \lambda_j(t,u,h,y)$ are health claim hazards of similar nature, which we also assume to be measurable and bounded on compacts. 

{
\begin{example}\label{ex:simple_g}
As a concrete example of how the function $g$ might be chosen in practice, consider $g(z,u,h) = \max\{h,\kappa_H\}$ for some $\kappa_H>0$. In this case,
\begin{align*}
\nu_t^n = \frac{1}{n} \sum_{\ell = 1}^n \max\{H_t^{\ell,n},\kappa_H\}.
\end{align*}
Consequently, the transition rate from the active to the disabled state may now depend on the average truncated number of health insurance claims. In practice, the impact of truncation can be removed by selecting $\kappa_H$ sufficiently large.
\end{example}
}

If the individuals have the same initial distribution, meaning $\pi^{\ell,n}$ does not depend on $\ell$, then~\eqref{eq:interacting_compensators} entails that the individuals are actually identically distributed. However, they are in general {not} independent and the Markov property holds neither for $X^{\ell,n}$ or $\nu^n$ nor $(X^{\ell,n}, \nu^n)$; only the high-dimensional process $X^n:=(X^{1,n},\ldots,X^{n,n})$ is always Markov. If we may think of $\nu^n$ as an external process driving $Z^{1,n}$, which in the context of Example~\ref{ex:simple_g} would be the case if $H^{\ell,n}$ is \textit{locally independent} of $Z^{\ell,n}$, confer with~\cite{Aalen1987}, computational simplifications could surface. It is the (causal) interaction between disabilities and health claims on a collective level that especially complicate the computational aspects. 

The fact that $X^n$ is a Markov process signifies, at least in principle, that transition probabilities and other quantities of interest may be computed by solving (high-dimensional) systems of integro-differential equations. However, even if we disregard health claims and duration effects, the computational complexity grows exponentially in $n$. Thus, already for moderate $n$ {such a direct approach likely} has to be abandoned. Monte Carlo methods constitute an alternative, but are rather slow and should, therefore, only be used as a kind of last resort.

In the next section, we explore an attractive alternative: mean-field approximations, where the process $\nu^n$ is replaced by its mean, restoring the independence between individuals as well as the computational tractability.

\section{Mean-field approximation}\label{sec:mean_field}

The idea behind mean-field approximations is to replace averages by mean values -- noting that as $n\to\infty$, the average of $n$ exchangeable random variables convergens to their mean. Specifically, we want to replace the averaged process $\nu^n$ in~\eqref{eq:interacting_compensators} by its mean. The advantage of such an approximation would be the replacement of the high-dimensional $n$-individual model by a one-individual limiting model, the so-called mean-field model.

\subsection{Setup and mean-field convergence}\label{sec:mean-field_setup}

The mean-field model corresponding to the $n$-individual model of Subsection~\ref{subsec:collective} is given by th{e t}rivariate process $\bar{X}=(\bar{Z},\bar{U},\bar{H})$, where $\bar{U}$ is the duration process associated with $\bar{Z}$ and $\bar{H}$ is a counting process, with the following characteristics. First, and foremost for notational convenience, we assume that $\bar{H}$ and $\bar{N}$, where $\bar{N}$ is the multivariate counting process associated with $\bar{Z}$, admit no simultaneous jumps. Further, we assume that 
\begin{equation}\label{eq:mean_field_compensators}
    \begin{aligned}
        \amsmathbb{E}[\bar{N}_{jk}(\mathrm{d}t) \, | \, \mathcal{F}_{t-}]  &= \mathds{1}_{\{\bar{Z}_{t-} = j\}} \mu_{jk}\big(t,\bar{U}_{t-},\bar{H}_{t-},v(t-)\big) \mathrm{d}t,\\
        \amsmathbb{E}[\bar{H}(\mathrm{d}t) \, | \, \mathcal{F}_{t-}]  &= \lambda_{\bar{Z}_{t-}}\big(t,\bar{U}_{t-},\bar{H}_{t-},v(t-)\big) \mathrm{d}t,
    \end{aligned}  
\end{equation}
where $t \mapsto v(t) := \amsmathbb{E}[g(\bar{X}_t)]$. It is apparent that the empirical average $\nu^n$ has been replaced by its expectation $v$. Instead of depending on other individuals, th{e c}ompensators of an individual now depend on the distribution of the process $\bar{X}$ through the mean $v$. {In other words, $\bar{X}=(\bar{Z},\bar{U},\bar{H})$ is a distribution-dependent trivariate process.} This has certain mathematical implications, some of which are beneficial and some of which are not, as we shall soon demonstrate. The type of convergence that lies behind a mean-field approximation relates to the notion of chaosticity{. In the mean-field literature, one usually employs a notion of chaosticity based on weak convergence. However, as discussed in Section~7 of~\cite{Hornung2025}, it can be preferable to use a stronger notion based on the total variation distance. In this manuscript, we exclusively use this stronger notion of chaosticity.

\begin{definition}\label{def:chaosticity}
Let $(S,\mathcal{S})$ be a measurable space, and let $\amsmathbb{Q}$ be a probability measure on $(S,\mathcal{S})$. Then a sequence of exchangeable probability measures $(\amsmathbb{Q}_n)_{n\in\amsmathbb{N}}$, with each $\amsmathbb{Q}_n$ defined on the corresponding product space $(S^n, \otimes_{\ell=1}^n \mathcal{S})$, is said to be \emph{$\amsmathbb{Q}$-chaotic} if
\begin{align*}
\forall k \in \amsmathbb{N} : \quad \amsmathbb{Q}_n^k \overset{\text{TV}}{\to} \otimes_{\ell = 1}^k \amsmathbb{Q},
\end{align*}
where $\amsmathbb{Q}_n^k(\,\cdot\,) := \amsmathbb{Q}_n(\, \cdot \times S^{n-k})$ for $k < n$. 
\end{definition}

\begin{remark}
By using the superscript $\text{TV}$, we signify convergence in the total variation distance. As a gentle reminder to the reader, we may define the total variation distance $d_{\text{TV}}$ on the set of probability measure $\mathcal{P}(S)$ of a measurable space $(S,\mathcal{S})$ according to
\begin{align*}
d_{\text{TV}}(\amsmathbb{Q}_1,\amsmathbb{Q}_2)
:=
\sup_{A\in\mathcal{S}}|\amsmathbb{Q}_1(A)-\amsmathbb{Q}_2(A)|,\quad \amsmathbb{Q}_1,\amsmathbb{Q}_2\in\mathcal{P}(S).
\end{align*}
An important technical observation is that $(\mathcal{P}(S),d_{\text{TV}})$ is complete.
\end{remark}
}

\begin{remark}
As {another} gentle reminder to the reader, if $\amsmathbb{Q}_n$ describes the distribution of random variables $(Y^{1,n},\ldots,Y^{n,n})$, then $\amsmathbb{Q}_n$ (or the random variables themselves) is said to be \textit{exchangeable} if
\begin{align*}
(Y^{1,n},\ldots,Y^{1,n}) \overset{\text{d}}{=} (Y^{\sigma(1),n},\ldots,Y^{\sigma(n),n})
\end{align*}
for all permutations $\sigma : \{1,\ldots,n\} \mapsto \{1,\ldots,n\}$. Intuitively, the joint distribution of the individuals is unchanged when the individuals are reordered and, consequently, all individuals must share the same marginal distribution. Note that independent and identically distributed random variables are automatically exchangeable.
\end{remark} 

Chaosticity, as Definition~\ref{def:chaosticity} reveals, entails that the individuals become asymptotically independent in the following manner: any fixed number of individuals become independent as the overall number of individuals goes to infinity. 

{In the following, we consider the processes $\bar{X}, X^{1,n}$, \ldots, $X^{n,n}$ as random variables taking values in the path space of non-explosive jump processes. This space can be equipped with the $\sigma$-algebra generated by the coordinate projections, and this $\sigma$-algebra corresponds to the Borel $\sigma$-algebra generated by the $J_1$-subspace topology.} We want to show that the sequence $(\amsmathbb{Q}_n)_{n\in\amsmathbb{N}}$ is $\bar{\amsmathbb{Q}}$-chaotic; this is what is understood by `mean-field convergence'.

So far, we have refrained from restricting the initial distribution of $X^n$ and $\bar{X}$. By definition, $U^{\ell,n}_0 = H^{\ell,n}_0 = \bar{U}_0 = \bar{H}_0 = 0$. The minimal condition would thus be to assume $\pi$-chaosticity of $(Z^{1,n}_0,\ldots,Z^{n,n}_0)(\amsmathbb{P})$ for some $\pi$, which is the assumption we adopt in the remainder of this section. Note that if the individuals are independent and identically distributed at time zero, then this minimal condition is automatically satisfied with $\pi = \pi^{1,1}$.

The following conditions on the transition rates, the health claim hazards, and the averaging function $g$ are sufficient to ensure the desired mean-field convergence. For just the existence and uniqueness of the mean-field model, less may do{.

\begin{condition}\label{cond:transition_rates_g}  
\begin{enumerate}[a)]
\item[]
\item The transition rates $(t,u,h,y) \mapsto \mu_{jk}(t,u,h,y)$, $j \neq k$, and the hazards $(t,u,h,y) \mapsto \lambda_j(t,u,h,y)$, in addition to being bounded on compacts, are bounded in $h$ and satisfy the following Lipschitz conditions: There exist non-negative constants $C_{jk}$ and $C_j$ such that, for all $y_1,y_2 \in \amsmathbb{R}^d$, all $u\in[0,T]$, and all $h \in \amsmathbb{N}_0$, it holds that
\begin{align*}
\lvert \mu_{jk}(t,u,h,y_1) - \mu_{jk}(t,u,h,y_2) \rvert &\leq C_{jk}\lVert y_1 - y_2 \rVert, \\
\lvert \lambda_j(t,u,h,y_1) - \lambda_j(t,u,h,y_2) \rvert &\leq C_j\lVert y_1 - y_2 \rVert.
\end{align*}
\item The averaging function $(z,u,h) \mapsto g(z,u,h)$, in addition to being bounded, satisfies the following Lipschitz condition: There exists a non-negative constant $C_g$ such that, for all $u_1,u_2\in[0,T]$, all $z \in \mathcal{J}$, and all $h\in\amsmathbb{N}_0$, it holds that
\begin{align*}
\lVert g(z,u_1,h) - g(z,u_2,h) \rVert \leq C_g \lvert u_1 - u_2 \rvert.
\end{align*}
\end{enumerate}
\end{condition}
}

{
\begin{remark}
Let $\mathcal{P}(E)$ denote the set of all probability measures on $E$ equipped with the Borel $\sigma$-algebra, and let $d_{\text{BL}}$ denote the so-called bounded-Lipschitz distance on $\mathcal{P}(E)$. Condition~\ref{cond:transition_rates_g} is sufficient to ensure that $(t,u,h,\rho) \mapsto \mu_{jk}(t,u,h,\int_E g \, \mathrm{d}\rho)$, $j \neq k$, and $(t,u,h,\rho) \mapsto \lambda_j(t,u,h,\int_E g \, \mathrm{d}\rho)$ are bounded and satisfy the following Lipschitz conditions: There exist non-negative constants $C_{jk}$ and $C_j$ such that, for all $\rho_1,\rho_2 \in \mathcal{P}(E)$, all $u\in [0,T]$, and all $h \in \amsmathbb{N}_0$, it holds that  
\begin{align*}
\bigg\lvert \mu_{jk}\Big(t,u,h,\int_E g \, \mathrm{d}\rho_1\Big) - \mu_{jk}\Big(t,u,h,\int_E g \, \mathrm{d}\rho_2\Big) \bigg\rvert &\leq C_{jk}d_{\text{BL}}(\rho_1,\rho_2), \\
\bigg\lvert \lambda_j\Big(t,u,h,\int_E g \, \mathrm{d}\rho_1\Big) - \lambda_j\Big(t,u,h,\int_E g \, \mathrm{d}\rho_2\Big) \bigg\rvert &\leq C_j d_{\text{BL}}(\rho_1,\rho_2).
\end{align*}
These are the Lipschitz conditions that are actually exploited to verify the existence and uniqueness of the mean-field model and to establish the desired mean-field convergence.
\end{remark}
}

Due to Condition~\ref{cond:transition_rates_g}, we have conveniently and firmly positioned ourselves within the general mean-field theory, see for instance~\cite{Hornung2025}. Consequently, we immediately obtain the following result, which confirms the desired mean-field convergence and establishes the mean-field model as an approximation to the $n$-individual model.
\begin{theorem}\label{thm:mean-field_convergence}
Suppose that Condition~\ref{cond:transition_rates_g} is met. Then the mean-field model exists and is unique. Further, with $\amsmathbb{Q}_n := X^n(\amsmathbb{P})$ and $\bar{\amsmathbb{Q}} := \bar{X}(\amsmathbb{P})$, it holds that $(\amsmathbb{Q}_n)_{n\in\amsmathbb{N}}$ is $\bar{\amsmathbb{Q}}$-chaotic.
\end{theorem}
\begin{proof}
We adopt the reparameterization of Remark~\ref{rmk:semi_markov_reparam} for both $X^{\ell,n}$ and $\bar{X}$. This yields the following systems of integral equations:
\begin{align*}
\mathrm{d}Z^{\ell,n}_t = \sum_k \big( k - Z^{\ell,n}_{t-}\big) N^{\ell,n}_k(\mathrm{d}t), \hspace{2mm} 
\mathrm{d}Y^{\ell,n}_t &= (t - Y^{\ell,n}_{t-}) \sum_k N^{\ell,n}_k(\mathrm{d}t), \hspace{2mm} 
\mathrm{d}H^{\ell,n}_t = \mathrm{d}H^{\ell,n}_t, \\
\mathrm{d}\bar{Z}_t = \sum_k \big( k - \bar{Z}_{t-}\big) \bar{N}_k(\mathrm{d}t), \hspace{2mm} 
\mathrm{d}\bar{Y}_t &= (t - \bar{Y}_{t-}) \sum_k \bar{N}_k(\mathrm{d}t), \hspace{2mm} 
\mathrm{d}\bar{H}_t = \mathrm{d}\bar{H}_t,
\end{align*}
where $N^{\ell,n}_k := \sum_j N^{\ell,n}_{jk}$ and $\bar{N}_k := \sum_j \bar{N}_{jk}$. Proposition~3.14 of~\cite{Hornung2025} shows that Condition~\ref{cond:transition_rates_g} is sufficient for the validity of Theorem~3.7 and Theorem~4.4 in~\cite{Hornung2025}, with the former ensuring the existence and uniqueness of the mean-field model and the latter confirming the postulated chaosticity.
\end{proof}

We conclude this subsection with a brief discussion on regular conditional distributions in mean-field models. If $Y$ is an ordinary Markov process, then it is completely described by its initial distribution and its transition probabilities. In particular, we may determine its occupation probabilities by integrating the transition probabilities with respect to the initial distribution or, vice versa, we may determine the initial distribution via disintegration of the occupation probabilities with respect to the transition probabilities. In other words, changing the initial distribution only affects the occupation probabilities -- not the transition probabilities; confer also with the discussions in Section~\ref{sec:pre_mean-field}.

In a mean-field model, the transition probabilities depend on the initial distribution, so changing the initial distribution affects not only the occupation probabilities, but also the transition probabilities. This means that caution should be exercised when interpreting conditional expectations such as $\amsmathbb{E}[h(\bar{X}) \, | \, \bar{Z}_0 = i]$ for some suitably regular function $h$. This expectation represents integration with respect to a regular conditional distribution of $\bar{X}$ given $\bar{Z}_0$, that is it involves the transition probabilities of $\bar{X}$, which depend on the limiting initial distribution $\pi$.

\subsection{Actuarial applications}

The contractual payments remain on the form~\eqref{eq:sec2:B}, meaning we consider $B^{\ell,n}$ and $\bar{B}$ given by
\begin{align*}
B^{\ell,n}(\mathrm{d}t)
&=
\sum_j \mathds{1}_{\{Z^{\ell,n}_{t-} = j\}} b_j(t,U^{\ell,n}_{t-}) \mathrm{d}t + \sum_{j \neq k} b_{jk}(t,U^{\ell,n}_{t-}) N^{\ell,n}_{jk}(\mathrm{d}t), \\ 
\bar{B}(\mathrm{d}t)
&=
\sum_j \mathds{1}_{\{\bar{Z}_{t-} = j\}} b_j(t,\bar{U}_{t-}) \mathrm{d}t + \sum_{j \neq k} b_{jk}(t,\bar{U}_{t-}) \bar{N}_{jk}(\mathrm{d}t).
\end{align*}
Recall that $(t,u) \mapsto b_j(t,u)$ and $(t,u) \mapsto b_{jk}(t,u)$ are measurable functions that are bounded on compacts{.

T}he following results provide an actuarial perspective on Theorem~\ref{thm:mean-field_convergence}, illuminating how the {total variation convergence} of state processes give{s} rise to laws of large numbers that substantiate the use of mean-field approximations for, among other things, { pricing and} reserving purposes.
{
\begin{proposition}\label{prop:L2_convergence}
Suppose that {Condition~\ref{cond:transition_rates_g} is} met. It then holds that
\begin{align*}
\frac{1}{n}\sum_{\ell=1}^n \int_0^T e^{-\int_0^t r(s) \, \mathrm{d}s}  B^{\ell,n}(\mathrm{d}t)
\overset{L^{\!2}}{\to}
\amsmathbb{E}\bigg[\int_0^T e^{-\int_0^t r(s) \, \mathrm{d}s}  \bar{B}(\mathrm{d}t) \bigg].
\end{align*}
Furthermore, if $\pi(i) > 0$, then
\begin{align*}
\frac{\frac{1}{n}\sum_{\ell=1}^n \mathds{1}_{\{Z^{\ell,n}_0 = i\}} \int_0^T e^{-\int_0^t r(s) \, \mathrm{d}s}  B^{\ell,n}(\mathrm{d}t)}
{\frac{1}{n} \sum_{\ell=1}^n \mathds{1}_{\{Z^{\ell,n}_0 = i\}}}
\overset{p}{\to}
\amsmathbb{E}\bigg[\int_0^T e^{-\int_0^t r(s) \, \mathrm{d}s}  \bar{B}(\mathrm{d}t) \, \bigg| \, \bar{Z}_0 = i \bigg].
\end{align*}
\end{proposition}
\begin{proof}
We can view $\bar{B}$ as a measurable mapping from the path space of non-explosive jump processes into $\amsmathbb{R}$. Theorem~\ref{thm:mean-field_convergence} ensures chaosticity and Proposition~6.10 of~\cite{Hornung2025}, Condition~\ref{cond:transition_rates_g}, and the boundedness  on compacts of $(t,u) \mapsto b_j(t,u)$ and $(t,u) \mapsto b_{jk}(t,u)$ ensure sufficient integrability. The desired results then follow from Proposition~6.6 of~\cite{Hornung2025}.
\end{proof}}  
The limits appearing in Proposition~\ref{prop:L2_convergence} {a}re so-called mean-field prospective reserves. The state-wise prospective mean-field prospective reserves at contract inception, $(\bar{V}_i)_i$, are defined according to
\begin{align*}
\bar{V}_i
=
\amsmathbb{E}\bigg[\int_0^T e^{-\int_0^t r(s) \, \mathrm{d}s} \bar{B}(\mathrm{d}t) \, \bigg| \, \bar{Z}_0 = i\bigg],
\end{align*}
while the portfolio-wide prospective mean-field reserve at contract inception $\bar{V}$ reads
\begin{align*}
\bar{V}
=
\amsmathbb{E}\bigg[\int_0^T e^{-\int_0^t r(s) \, \mathrm{d}s} \bar{B}(\mathrm{d}t) \bigg].
\end{align*}
If we denote by $(\bar{A}_i)_i$ the so-called state-wise mean-field accumulated cash flows, which are given by
\begin{align*}
\bar{A}_i(\mathrm{d}t) = \amsmathbb{E}[\bar{B}(\mathrm{d}t) \, | \, \bar{Z}_0 = i], 
\end{align*}
then it holds that
\begin{align*}
\bar{V}_i
=
\int_0^T e^{-\int_0^t r(s) \, \mathrm{d}s} \bar{A}_i(\mathrm{d}t).
\end{align*} 
In similar fashion, with $\bar{A}$ the mean-field accumulated cash flow given by 
\begin{align*}
\bar{A}(\mathrm{d}t) = \amsmathbb{E}[\bar{B}(\mathrm{d}t)],
\end{align*}
it holds that
\begin{align*}
\bar{V}
=
\int_0^T e^{-\int_0^t r(s) \, \mathrm{d}s} \bar{A}(\mathrm{d}t).
\end{align*}
These prospective reserves and expected cash flows should be seen in comparison to those of the $n$-individual model, namely $(V^{1,n}_i)_i$, $V^{1,n}$, $(A^{1,n}_i)_i$, and $A^{1,n}$ given by
\begin{align*}
V^{1,n}_i
=
\amsmathbb{E}\bigg[\int_0^T e^{-\int_0^t r(s) \, \mathrm{d}s} B^{1,n}(\mathrm{d}t) \, \bigg| \, Z^{1,n}_0 = i\bigg],
&\quad
V^{1,n}
=
\amsmathbb{E}\bigg[\int_0^T e^{-\int_0^t r(s) \, \mathrm{d}s} B^{1,n}(\mathrm{d}t) \bigg], \\
A^{1,n}_i(\mathrm{d}t) = \amsmathbb{E}[B^{1,n}(\mathrm{d}t) \, | \,Z^{1,n}_0 = i],
&\quad
A^{1,n}(\mathrm{d}t) = \amsmathbb{E}[B^{1,n}(\mathrm{d}t)],
\end{align*}
and satisfying
\begin{align*}
V^{1,n}_i
=
\int_0^T e^{-\int_0^t r(s) \, \mathrm{d}s} A^{1,n}_i(\mathrm{d}t), \quad
V^{1,n}
=
\int_0^T e^{-\int_0^t r(s) \, \mathrm{d}s} A^{1,n}(\mathrm{d}t).
\end{align*}
The following proposition establishes the mean-field cash flows and reserves {as} viable approximations of their $n$-individual counterparts.

{
\begin{proposition}
Suppose that Condition~\ref{cond:transition_rates_g} is met. It then holds that
\begin{align*}
V^{1,n} \to \bar{V} \quad \text{and} \quad \forall t\geq0: A^{1,n}(t) - A^{1,n}(0) \to \bar{A}(t) - \bar{A}(0).
\end{align*}
Furthermore, if $\pi(i)>0$,, then
\begin{align*}
 V^{1,n}_i \to \bar{V}_i \quad \text{and} \quad \forall t\geq0: A^{1,n}_i(t) - A^{1,n}_i(0) \to \bar{A}_i(t) - \bar{A}_i(0).
\end{align*}
\end{proposition}
\begin{proof}
The argument for the expected cash flows is similar to that of the reserves, so we focus on the latter. Theorem~\ref{thm:mean-field_convergence} ensures chaosticity and Proposition~6.10 of~\cite{Hornung2025}, Condition~\ref{cond:transition_rates_g}, and the boundedness on compacts of $(t,u) \mapsto b_j(t,u)$ and $(t,u) \mapsto b_{jk}(t,u)$ ensure sufficient integrability. The desired result then follows from Proposition~6.3 and Proposition~6.4 of~\cite{Hornung2025}.
\end{proof}} 

Having formally verified the usefulness of the mean-field model, we now discuss how to calculate mean-field cash flows and reserves. To this end, let
\begin{align*}
\bar{p}_{ij}(t,u,h) :=&\, \amsmathbb{P}(\bar{Z}_t = j, \bar{U}_t \leq u, \bar{H}_t=h \, | \, \bar{Z}_0 = i), \\
\bar{p}_j(t,u,h) :=&\, \amsmathbb{P}(\bar{Z}_t = j, \bar{U}_t \leq u,\bar{H}_t=h)
\end{align*}
be the mean-field transition probabilities and occupation probabilities, respectively. Note that, with $\pi$ the limiting initial distribution,
\begin{align}\label{eq:sum_formula}
\bar{p}_j = \sum_i \pi(i) \bar{p}_{ij}, \quad \bar{A} = \sum_i \pi(i) \bar{A}_i, \quad \bar{V} = \sum_i \pi(i) \bar{V}_i.
\end{align}

The followin{g p}roposition expands Proposition~\ref{prop:semi_markov_cf} to not only include individual health claims, but also mean-field effects. It should be compared to Proposition~\ref{prop:semi_markov_health_cf}.

\begin{proposition}\label{prop:mean-field_cf}
It holds that
\begin{align*}
\bar{A}_i(\mathrm{d}t)
&=
\sum_j\sum_{h=0}^{\infty}\int_0^t \Big(b_j(t,u)+ \sum_{k : k \neq j} b_{jk}(t,u) \mu_{jk}\big(t,u,h,v(t)\big) \Big)  \bar{p}_{ij}(t,\mathrm{d}u,h) \, \mathrm{d}t, \\
\bar{A}(\mathrm{d}t)
&=
\sum_j \sum_{h=0}^{\infty} \int_0^t \Big(b_j(t,u) + \sum_{k : k \neq j} b_{jk}(t,u) \mu_{jk}\big(t,u,h,v(t)\big)\Big)  \bar{p}_j(t,\mathrm{d}u,h) \, \mathrm{d}t, \\
v(t)
&=
\sum_j \sum_{h=0}^\infty \int_0^t g(j,u,h) \bar{p}_j(t,\mathrm{d}u,h)
=
\sum_i \pi(i) \sum_j \sum_{h=0}^\infty \int_0^t g(j,u,h) \bar{p}_{ij}(t,\mathrm{d}u,h).
\end{align*}
\end{proposition}
\begin{proof}
The core argument is exactly the same as in the proof of Proposition~\ref{prop:semi_markov_health_cf}.
\end{proof}

It is worth noting that the calculation of mean-field accumulated cash flows, and therefore also mean-field reserves, is thus no more involved than the calculation of expected accumulated cash flows and reserves in the one-individual model -- if occupation and transition probabilities are readily available and intermediaries such as $v$ are easy to calculate.

\begin{example}\label{ex:simple_g_continued}
{Continuing Example~\ref{ex:simple_g} with $g(z,u,h) = \max\{h,\kappa_H\}$, it follows that
\begin{align*}
v(t) = \sum_j \sum_{h=0}^\infty \max\{h,\kappa_H\} \bar{p}_j(t,t,h) = \sum_i \pi(i) \sum_j \sum_{h=0}^\infty \max\{h,\kappa_H\} \bar{p}_{ij}(t,t,h),
\end{align*}
which is a rather straightforward to calculate based on the mean-field occupation or transition probabilities.}
\end{example}

Let us begin by exploring the calculation of the mean-field occupation probabilities. In the following, we adopt the convention that $(t,u,v) \mapsto \lambda_j(t,u,-1,v)$ is constantly zero.

\begin{proposition}\label{prop:forward_eq_occ_mean-field}
Let $d\geq0$. It holds almost everywhere on $[d,\infty)$ that
\begin{align*}
&\frac{\mathrm{d}}{\mathrm{d}t} \bar{p}_{j}(t,t-d,h) \\
&=
\sum_{k : k \neq j} \int_0^t \mu_{kj}\big(t,u,h,v(t)\big) \bar{p}_k(t,\mathrm{d}u,h)-\int_0^{t-d} \mu_{j\bigcdot}\big(t,u,h,v(t)\big) \bar{p}_{j}(t,\mathrm{d}u,h)\\
&\quad+
\int_{0}^{t-d}\lambda_j\big(t,u,h-1,v(t)\big)\bar{p}_{j}(t,\mathrm{d}u,h-1)-\int_0^{t-d}\lambda_j\big(t,u,h,v(t)\big) \bar{p}_{j}(t,\mathrm{d}u,h) 
\end{align*}
with boundary conditions $\bar{p}_j(t,0,h) = \mathds{1}_{\{t=0\}}\mathds{1}_{\{h=0\}}\pi(j)$.
\end{proposition}
\begin{proof}
Adopting the reparameterization and change of variables from Remark~\ref{rmk:semi_markov_reparam}, the result follows from Proposition~3.9 of~\cite{Hornung2025}, which yields a forward equation for $(\bar{Z},\bar{Y},\bar{H})$.
\end{proof}

Contrary to the second part of Proposition~\ref{prop:semi_markov_forward_health}, the forward equation for the mean-field occupation probabilities is non-linear since the transition rates and health claim hazards depend on $v$ (in a possibly non-linear fashion) and $v$ is a function of the occupation probabilities themselves, confer with Proposition~\ref{prop:mean-field_cf}. As a consequence, we can no longer be completely assured that the equations admit a unique solution. The practical consequence is that the actuary must be particularly vigilant regarding a solution's mathematical characteristics and concrete impact.

Furthermore, the non-linearity of the equations gives reason to handle the mean-field transition probabilities with additional care. In the one-individual model, confer with Remark~\ref{rmk:occ_to_trans}, the evolution in time of the occupation and transition probabilities was identical and, hence, we could calculate the transition probabilities using the same equations as for the occupation probabilities, but with changed initial conditions (from $\pi(j)$ to $\mathds{1}_{\{i = j\}}$). However, as we discussed briefly at the end of Subsection~\ref{sec:mean-field_setup}, in the mean-field model the evolution in time depends on the initial distribution through $v$. The intuition is as follows. The mean-field occupation probabilities $(\bar{p}_j)_j$ represent the occupation probabilities for a typical individual in a very large group of identically distributed individuals, who \textit{all} have initial distribution $\pi$ and all depend on each other through the group average. Thus changing the initial distribution $\pi$ not only corresponds to changing the initial distribution of one individual, but it also corresponds to changing the initial distribution of the \textit{entire} group and thus the group average. Or, in other words, trying to equate occupation probabilities from equations with different initial conditions corresponds to equating the occupation probabilities of two individuals with different initial conditions, but also from two different groups!

In summary, we \textit{cannot} calculate the mean-field transition probabilities by simply changing the initial (or boundary) conditions. However, if we treat $v$, which depends on the occupation probabilities, as fixed, we may actually calculate the mean-field transition probabilities in the usual manner.

\begin{proposition}\label{prop:forward_eq_trans_mean-field}
Let $d\geq0$. It holds almost everywhere on $[d,\infty)$ that
\begin{align*}
&\frac{\mathrm{d}}{\mathrm{d}t} \bar{p}_{ij}(t,t-d,h) \\
&=
\sum_{k : k \neq j} \int_0^t \mu_{kj}\big(t,u,h,v(t)\big) \bar{p}_{ij}(t,\mathrm{d}u,h)-\int_0^{t-d} \mu_{j\bigcdot}\big(t,u,h,v(t)\big) \bar{p}_{ij}(t,\mathrm{d}u,h)\\
&\quad+
\int_{0}^{t-d}\lambda_j\big(t,u,h-1,v(t)\big)\bar{p}_{ij}(t,\mathrm{d}u,h-1)-\int_0^{t-d}\lambda_j\big(t,u,h,v(t)\big) \bar{p}_{ij}(t,\mathrm{d}u,h) 
\end{align*}
with boundary conditions $\bar{p}_{ij}(t,0,h) = \mathds{1}_{\{t=0\}}\mathds{1}_{\{h=0\}}\mathds{1}_{\{i = j\}}$.
\end{proposition}
\begin{proof}
Adopting the reparameterization and change of variables from Remark~\ref{rmk:semi_markov_reparam}, the result follows from Proposition~3.11 of~\cite{Hornung2025}, which yields a forward equation for $(\bar{Z},\bar{Y},\bar{H})$.
\end{proof}

The intuition behind these linearized forward equations is as follows. Imagine a very large group of $M$ individuals of which all but one have initial distribution $\pi$, while the remaining individual has a degenerate initial distribution, say this individual's initial state is almost surely $i$. Since the individuals solely depend on each other through their group average, and the contribution of one individual to this average is negligible as $M \to \infty$, we may still replace the average by $v$ -- also when calculating the occupation probabilities of the remaining individual. However, for this individual the occupation probabilities now correspond to the mean-field transition probabilities $(\bar{p}_{ij})_j$.

Collecting results leaves us with two ways of calculating the mean-field transition probabilities. Either we first determine the mean-field occupation probabilities by solving the non-linear forward equations from Proposition~\ref{prop:forward_eq_occ_mean-field}, use these to calculate $v$, and then finally find the mean-field transition probabilities by solving the linearized forward equations from Proposition~\ref{prop:forward_eq_trans_mean-field}.  Alternatively, we recall that
\begin{align*}
v(t)
=
\sum_i \pi(i) \sum_j \sum_{h=0}^\infty \int_0^t g(j,u,h) \bar{p}_{ij}(t,\mathrm{d}u,h),
\end{align*}
consider the forward equations from Proposition~\ref{prop:forward_eq_trans_mean-field} as non-linear, and solve these directly. If one is only interested in the mean-field transition probabilities $(\bar{p}_{ij})_j$ for a specific $i$, the first method is to be preferred. Otherwise, neither method is inherently superior to the other and, in any case, for both methods the non-linearity entails that we may no longer be absolutely certain that the resulting solution is unique. 

\subsection{Statistical aspects}

For most practical purposes, estimates of the (collective-dependent) health claims hazards $(t, u, h, y) \mapsto \lambda_j(t,u,h,y)$ and transition rates $(t,u,h,y) \mapsto \mu_{jk}(t,u,h,y)$, $j \neq k$, are required. If only a single collective is observed, identifiability of the collective effect may become particularly challenging. However, as briefly described in the introduction , we have the example of company level insurance plans in mind -- with the insurer signing contracts with several (relatively independent) companies. In the following, we therefore outline how estimates may be obtained in the presence of multiple, mutually independent, groups; for notational convenience, we omit the inclusion of individual- and company-level covariates.

We begin by considering a single company consisting of $n$ employees observed in the interval $[0,R^n]$, with $R^n$ a common random time describing right-censoring of the company. Subject to classic assumptions, including \textit{independent right-censoring}, the partial log-likelihoods
\begin{align*}
\log\{\mathcal{L}\}
&=
\sum_{\ell=1}^n \int_0^{R^n} \log\!\big\{\lambda_{Z^{\ell,n}_{t-}}\big(t,U_{t-}^{\ell,n},H_{t-}^{\ell,n},\nu_{t-}^n\big)\big\} H^{\ell,n}(\mathrm{d}t) \\
&\quad -
\sum_{\ell=1}^n \int_0^{R^n} \lambda_{Z^{\ell,n}_{t-}}\big(t,U_{t-}^{\ell,n},H_{t-}^{\ell,n},\nu_{t-}^n\big)  \, \mathrm{d}t, \\
\log\{\mathcal{L}_{jk}\}  
&=
\sum_{\ell=1}^n \int_0^{R^n} \log\!\big\{\mu_{jk}\big(t,U_{t-}^{\ell,n},H_{t-}^{\ell,n},\nu_{t-}^n\big)\big\} N^{\ell,n}_{jk}(\mathrm{d}t) \\
&\quad -
\sum_{\ell=1}^n \int_0^{R^n} \mu_{jk}\big(t,U_{t-}^{\ell,n},H_{t-}^{\ell,n},\nu_{t-}^n\big) \mathds{1}_{\{Z^{\ell,n}_{t-} = j\}} \, \mathrm{d}t, \\
\end{align*}
offer a reasonable starting point for inference, confer with Section~III.4 in~\cite{AndersenBorganGillKeiding1993}.

In the presence of multiple, mutually independent, companies, the relevant partial log-likelihoods are simply sums of the each company's contribution. Temporarily discretizing the transition rates and health claims hazards using a grid with time, duration, etc., usually produces a good approximation, and the resulting expressions correspond to Poisson likelihoods with occurrences and exposures as one might expect. Therefore, estimates of health claims hazards and transition rates may be obtained non-parametrically, semi-parametrically, or parametrically using standard techniques for occurrence and exposure data. 

\section{Practical implementation}\label{sec:practical}

The expected accumulated cash flows and prospective reserves may be computed from the transition and occupation probabilities via numerical integration and, for instance, using the trapezoidal rule. It is the computaton of probabilities based on forward integro-differential equations that requires special attention. In the following, we briefly describe how the meta-algorithm of~\cite{BuchardtMollerSchmidt2015} can be adapted to be fit for purpose for the task at hand. We focus on the system of Proposition~\ref{prop:forward_eq_occ_mean-field}; the other systems of integro-differential equations are, ultimately, special cases or of significantly less sophistication. 

Obviously, $\bar{p}_j(t,t-d,h) = 0$ for $d>t$ and $\bar{p}_j(t,t-d,h) = \bar{p}_j(t,t,h)$ for $d < 0$. We therefore for $\eta > 0$ with $T/\eta \in \amsmathbb{N}$ consider the discretization $\mathcal{D}$ of $\{(t,d,h) \in [0,T]^2 \times \amsmathbb{N}_0  : d \leq t\}$ consisting of points $(\eta m, \eta n, h)$ for $n,m,h \in \amsmathbb{N}_0$ with $n \leq m \leq T/\eta$.

The goal is to calculate $(\bar{p}_j)_j$ on $\mathcal{D}$. This first involves selecting a cut-off $K_H$ and for all $j$ equating $\bar{p}_j(\cdot,\cdot,h)$ with zero for $h>K_H$. To select the cut-off, one may look for a deterministic constant $\tilde{\lambda}$ which uniformly bounds the health claims hazards on $[0,T]$, and then for an error threshold $\text{err}>0$ select $K_H = \inf\{K \in \amsmathbb{N}_0 : \amsmathbb{P}(\tilde{H} > K) < \text{err}\}$, where $\tilde{H} \sim \text{Poisson}(\tilde{\lambda} T)$. Next, one can apply the following meta-algorithm:

\begin{adjustwidth}{5mm}{5mm}

\textbf{Initial stage ($0$).} The boundary conditions yield the values $\bar{p}_j(0,0,h) = \mathds{1}_{\{h=0\}}\pi(j)$ for all $j$ and all $h$. 

\textbf{Subsequent stages ($m+1$).} The non-linear integro-differential equations of Proposition~\ref{prop:forward_eq_occ_mean-field} together with the formula for $v$ of Proposition~\ref{prop:mean-field_cf} may be used to calculate
\begin{align*}
&\bar{p}_j(\eta (m+1),\eta (m+1) - d,h) \\ &\text{for all } j \text{ and all } h \text{ and } d\in\{0,\eta,\eta2,\ldots,\eta (m+1)\}
\end{align*}
based on 
\begin{align*}
&\bar{p}_j(\eta m,\eta m - d,h) \text{ for all } j \text{ and all } h \text{ and } d\in\{0,\eta,\eta2,\ldots,\eta m\}
\end{align*}
and the boundary conditions
\begin{align*}
&\bar{p}_j(\eta (m+1),0,h) = 0 \text{ for all } j \text{ and all } h. 
\end{align*}
This can be done using for instance Euler steps, taking care of the integrals via the trapezoidal rule. Efficiency may be gained by storing and reusing computations related to the integrals.
\end{adjustwidth}

The sequence of stages for $m=0,1,2,3$ is illustrated in Figure~\ref{fig:num_scheme}, which mirrors Figure~2 in~\cite{BuchardtMollerSchmidt2015}. The time complexity given a cut-off $K_H$ is of the same order as in the classic semi-Markov disability model, but differs by about a factor of $K_H$. 

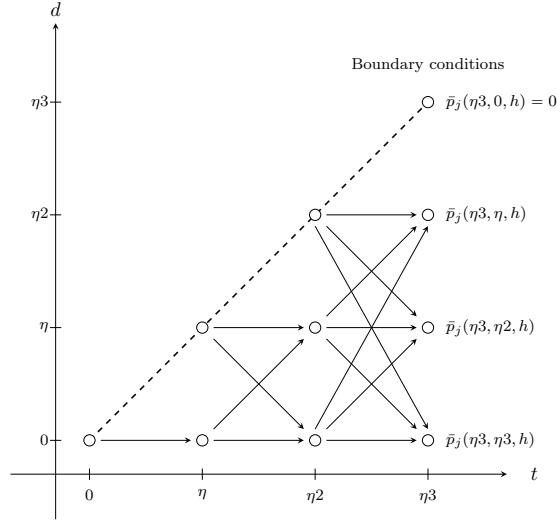
\begin{figure}[ht!]
    \centering
    \scalebox{0.75}{
\begin{tikzpicture}[scale=2, >=stealth]

\draw[->] (-0.2,0.2) -- (4.2,0.2) node[right=0.3cm] {$t$};
\draw[->] (0.2,-0.2) -- (0.2,4.2) node[above] {$d$};

\foreach \x in {0.5,1.5,2.5,3.5} {
    \draw[shift={(\x,0.2)}] (0,0.05) -- (0,-0.05);
}
\foreach \y in {0.5,1.5,2.5,3.5} {
    \draw[shift={(0.2,\y)}] (0.05,0) -- (-0.05,0);
}

\draw[dashed, black, thick] (0.5, 0.5) -- (3.5, 3.5);
\node[above] at (3.5,3.7) {\footnotesize Boundary conditions};

\draw[fill=white] (0.5,0.5) circle(0.05);  
\draw[fill=white] (1.5,0.5) circle(0.05);  
\draw[fill=white] (2.5,0.5) circle(0.05);  
\draw[fill=white] (3.5,0.5) circle(0.05);  
\draw[fill=white] (1.5,1.5) circle(0.05);  
\draw[fill=white] (2.5,1.5) circle(0.05);  
\draw[fill=white] (3.5,1.5) circle(0.05);  
\draw[fill=white] (2.5,2.5) circle(0.05);  
\draw[fill=white] (3.5,2.5) circle(0.05);  
\draw[fill=white] (3.5,3.5) circle(0.05);  

\draw[black, ->] (0.6,0.5) -- (1.4,0.5); 
\draw[black, ->] (1.6,1.5) -- (2.4,1.5); 
\draw[black, ->] (1.6,0.6) -- (2.4,1.4); 
\draw[black, ->] (1.6,0.5) -- (2.4,0.5); 
\draw[black, ->] (1.6,1.4) -- (2.4,0.6); 
\draw[black, ->] (2.6,2.5) -- (3.4,2.5); 
\draw[black, ->] (2.6,1.6) -- (3.4,2.4); 
\draw[black, ->] (2.5,0.6) -- (3.5,2.4); 
\draw[black, ->] (2.6,2.4) -- (3.4,1.6); 
\draw[black, ->] (2.6,1.5) -- (3.4,1.5); 
\draw[black, ->] (2.6,0.6) -- (3.4,1.4); 
\draw[black, ->] (2.5,2.4) -- (3.5,0.6); 
\draw[black, ->] (2.6,1.4) -- (3.4,0.6); 
\draw[black, ->] (2.6,0.5) -- (3.4,0.5); 

\node[black, right] at (3.6,3.5) {\footnotesize $\bar{p}_j(\eta 3, 0, h) = 0$};
\node[black, right] at (3.6,2.5) {\footnotesize $\bar{p}_j(\eta 3, \eta, h)$};
\node[black, right] at (3.6,1.5) {\footnotesize $\bar{p}_j(\eta 3, \eta2, h)$};
\node[black, right] at (3.6,0.5) {\footnotesize $\bar{p}_j(\eta 3, \eta3, h)$};

\node[below=0.15cm] at (0.5,0.2) {\footnotesize $0$};
\node[below=0.15cm] at (1.5,0.2) {\footnotesize $\eta$};
\node[below=0.15cm] at (2.5,0.2) {\footnotesize $\eta2$};
\node[below=0.15cm] at (3.5,0.2) {\footnotesize $\eta3$};
\node[left] at (0.2,0.5) {\footnotesize $0$};
\node[left] at (0.2,1.5) {\footnotesize $\eta$};
\node[left] at (0.2,2.5) {\footnotesize $\eta2$};
\node[left] at (0.2,3.5) {\footnotesize $\eta3$};

\end{tikzpicture}}
    \caption{Illustration of a sequence of stages in the meta-algorithm.}
    \label{fig:num_scheme}
\end{figure}

\section{Simulation study}\label{sec:sim_study}

The purpose of this section is to assess the quality of the mean-field approximation through a practice-oriented simulation study. In Subsection~\ref{sec:sim_setup}, we outline the specific model and its connection to practice. Subsection~\ref{sec:sim_reserve} concerns the quality of mean-field reserves compared to a naïve Monte Carlo approach{, while} Subsection~\ref{sec:sim_findings} presents some additional insights on the nature and the convergence of the mean-field approximation. {Finally, in Subsection~\ref{sec:lapse} we discuss how to adapt the model to also handle lapse risk -- from a practical perspective.}

\subsection{Setup} \label{sec:sim_setup}

In order to specify a concrete model, we must specify the initial distribution, transition rates $(t,u,h,y) \mapsto \mu_{jk}(t,u,h,y)$, $j \neq k$, health claims hazards $(t, u, h, y) \mapsto \lambda_j(t,u,h,y)$, and the function $g$. For the initial distribution, we assume for simplicity that individuals are independent and active at time zero. For the health claims hazards, we follow Example~\ref{ex:health_claims} and choose $\lambda_j(t,u,h,y)\equiv \lambda_j$ with parameter values
\begin{align*}
\lambda_1 = 0.2, \quad \lambda_2 = 0.3, \quad \lambda_3 = 0,
\end{align*}
implying in particular that health claims are $50\, \%$ more likely while disabled than active. For the transition rates, we choose
\begin{align*}
    \mu_{12}(t,y)&=e^{-9.55+0.24(t+45) - 0.0046(t + 45)^2 + 0.000036(t + 45)^3} e^{\beta\min\!\big\{\frac{1}{1+t}(y+\zeta_1)-\zeta_1,\zeta_0\big\}}, \\
    \mu_{13}(t)&=0.0005 + 10^{5.52 + 0.038(t + 45) - 10}, \\
    \mu_{21}(t,u)&=e^{2.11 - 0.039  (t + 45) - 1.44  u}, \\
    \mu_{23}(t,u)&=0.0005 + 10^{5.52 + 0.038(t + 45) - 10} + e^{-2.79 - 0.23  u}.
\end{align*}
In particular, solely the disability rate depends on the collective and solely the recovery rate and the disability mortality depend on duration. The collective effect on the disability rate is included via the term
\begin{align*}
\beta\min\!\Big\{\frac{1}{1+t}(y+\zeta_1)-\zeta_1,\zeta_0\Big\}
\end{align*}
with parameter values
\begin{align*}
\beta = 2, \quad \zeta_1 = 0.1, \quad \zeta_0 = 0.4.
\end{align*}
Finally, we follow Examples~\ref{ex:simple_g} and~\ref{ex:simple_g_continued} and choose {$g(z,u,h) \equiv \max\{h,\kappa_H\}$ with $\kappa_H = 100$, meaning that the dependence on the collective stems practically speaking} only from the average of health claims. This corresponds to the average and expectation 
\begin{align*}
\nu^n_t = \frac{1}{n} \sum_{\ell=1}^n {\max\{H^{\ell,n}_t,\kappa_H\}}, \quad v(t) = \sum_j \sum_{h=0}^\infty {\max\{h,\kappa_H\}} \bar{p}_{1j}(t,t,h)
\end{align*}
for the $n$-individual and mean-field model, respectively.

The collective health claims influence the disability rate $\mu_{12}$ by means of a credibility factor, taking into account time passed. Inserting $y = \nu^n$, we identify the term
\begin{align*}
\frac{1}{1+t}(\nu^n_t+\zeta_1) = \frac{t}{t+1} \frac{\nu^n_t}{t} + \frac{1}{t+1} \zeta_1
\end{align*}
as a credibility formula between the collective rate $\nu^n_t / t$ and a baseline $\zeta_1$. Consequently,
\begin{align*}
\frac{1}{1+t}(\nu^n_t+\zeta_1) - \zeta_1
\end{align*}
yields the deviation from the baseline. At time zero, where no collective information is available, all weight is placed on the baseline $\zeta_1$. As time passes -- and more and more information becomes available -- more and more weight is based on the collective rate $\nu^n_t / t$. However, by introducing a maximum positive deviation given by the parameter $\zeta_0$, we ensure that in no case can the deviation exceed $\zeta_0${. F}inally, the parameter $\beta$ controls the influence of health claims on the disability rate.

The parametrizations mirror forms seen in practice, and the parameter values are chosen to obtain rates which are reasonable for an individual of age $45$ years. Select parameter values are collected in Table~\ref{tbl:params}. It is worth noting that the baseline $\zeta_1$ is chosen quite low compared to the health claims rates $\lambda_1$, $\lambda_2$, $\lambda_3$.

  \begin{table}[h!]
\begin{tabular}{lrl}
\hline
Parameter		& \multicolumn{1}{l}{Value}	& Description \\ \hline
$\lambda_1$	& 0.2						& Health claims rate, active\\
$\lambda_2$	& 0.3						& Health claims rate, disabled\\
$\lambda_3$	& 0							& Health claims rate, dead\\
$\zeta_1$		& 0.1						& Collective health claims effect, baseline\\
$\zeta_0$		& 0.5						& Collective health claims effect, maximum \\
$\beta$		& 2							& Collective health claims effect, influence \\  
\hline
\end{tabular}
\caption{Select parameter values for the simulation study.}
\label{tbl:params}
\end{table}

Following Example~\ref{ex:waiting_period_continued}, the contractual payments correspond to a disability annuity with a waiting period. That is, we consider
\begin{align*}
\mathrm{d}B_t^{\ell,n}=\mathds{1}_{(Z_{t-}^{\ell,n}=2)}b_2(t,U_{t{-}}^{\ell,n})\mathrm{d}t
\end{align*}
with $b_2(t,u)=\mathds{1}_{(u\geq \varepsilon)}b$. We choose $b = 1$ and $\varepsilon = 0.25$, the latter corresponding to the rather common waiting period of three months. Finally, we choose $r = 0.01$ and $T = 25$.

\subsection{Main results} \label{sec:sim_reserve}

Since solving the forward integro-differential equations for the $n$-individual model is not computationally feasible for $n \gg 1$, we mainly consider two ways of calculating the reserve $V^{1,n}$:  
\begin{enumerate}
\item By means of a naïve Monte Carlo method, repeatedly simulating the $n$-individual model
\item Employing the mean-field approximation $\bar{V} \approx V^{1,n}$ and solving the resulting non-linear integro-differential equations.
\end{enumerate}
In the naïve Monte Carlo method, we repeatedly sample paths of the process $X^n = (X^{1,n},\ldots,X^{n,n})$ via inhomogeneous Poisson processes using the by now classic acceptance-rejection method described in~\cite{LewisShedler1979}. Denoting the samples by $m = 1, \ldots, M$, this yields the estimate
\begin{align*}
\frac{1}{M} \sum_{m=1}^M \bigg(\frac{1}{n}\sum_{\ell=1}^n \int_0^T e^{-\int_0^t r(s) \, \mathrm{d}s} B^{\ell,n,m}(\mathrm{d}t)\bigg)\!,
\end{align*}
which for $n \gg {1}$ should have substantially lower variance than the neonatal estimate
\begin{align*}
\frac{1}{M} \sum_{m=1}^M \int_0^T e^{-\int_0^t r(s) \, \mathrm{d}s} B^{1,n,m}(\mathrm{d}t)
\end{align*}
due to chaosticity. 

To solve the integro-differential equations, compute the expected accumulated cash flows, and finally calculate the reserves, we {adopt} the implementation outlined in Section~\ref{sec:practical}. In particular, we {employ} the Euler method and, for numerical integration, the trapezoidal rule. {The numerical methods are implemented in the programming language R. For the choice of cut-off, $K_H = 20$ is more than sufficient since the health claims hazards are bounded by $\tilde{\lambda}=0.3$ and with $\tilde{H} \sim \text{Poisson}(\tilde{\lambda}T)$ the excess probability $\amsmathbb{P}(\tilde{H} > 20)$ is negligible. In Figure~\ref{fig:conv}, we illustrate the convergence of the numerical method in terms of the total number of steps and the cut-off $K_H$. The top plot shows the computed reserve for a fixed cut-off $K_H = 20$ and various total number of steps. Based hereon, we select a step length of $0.0125$ corresponding to a total number of steps of $2,000$. The bottom plot shows the computed reserve for this step length and various cut-offs. Based hereon, we select a cut-off of $K_H = 15$.}

For $n > 1$ we only consider the mean-field approximation. For $n = 1$ we also consider the one-individual model, yielding the `true' value for the reserve. We present our main results in Table~\ref{tbl:main}; for the Monte Carlo method, a large sample size of $M = 40,000$ is chosen. The first observation is that that both the mean-field approximation and the Monte Carlo estimate deviate substantially from the true value for the one-individual model. While this is to be expected from the mean-field approximation, and indicates that the collective effect on the disability rate is non-negligible, it also signifies that the Monte Carlo estimate has not converged for $n = 1$. For $n=2,5$ we continue to see substantial differences between the mean-field approximation and the Monte Carlo estimates, while for $n=25,50,100$ the differences are less pronounced. The seemingly increased stability of the Monte Carlo estimates for larger $n$ are due to the aforementioned chaosticity-induced variance reduction.

\begin{figure}[h!]
    \centering
    \includegraphics[width=0.9\textwidth]{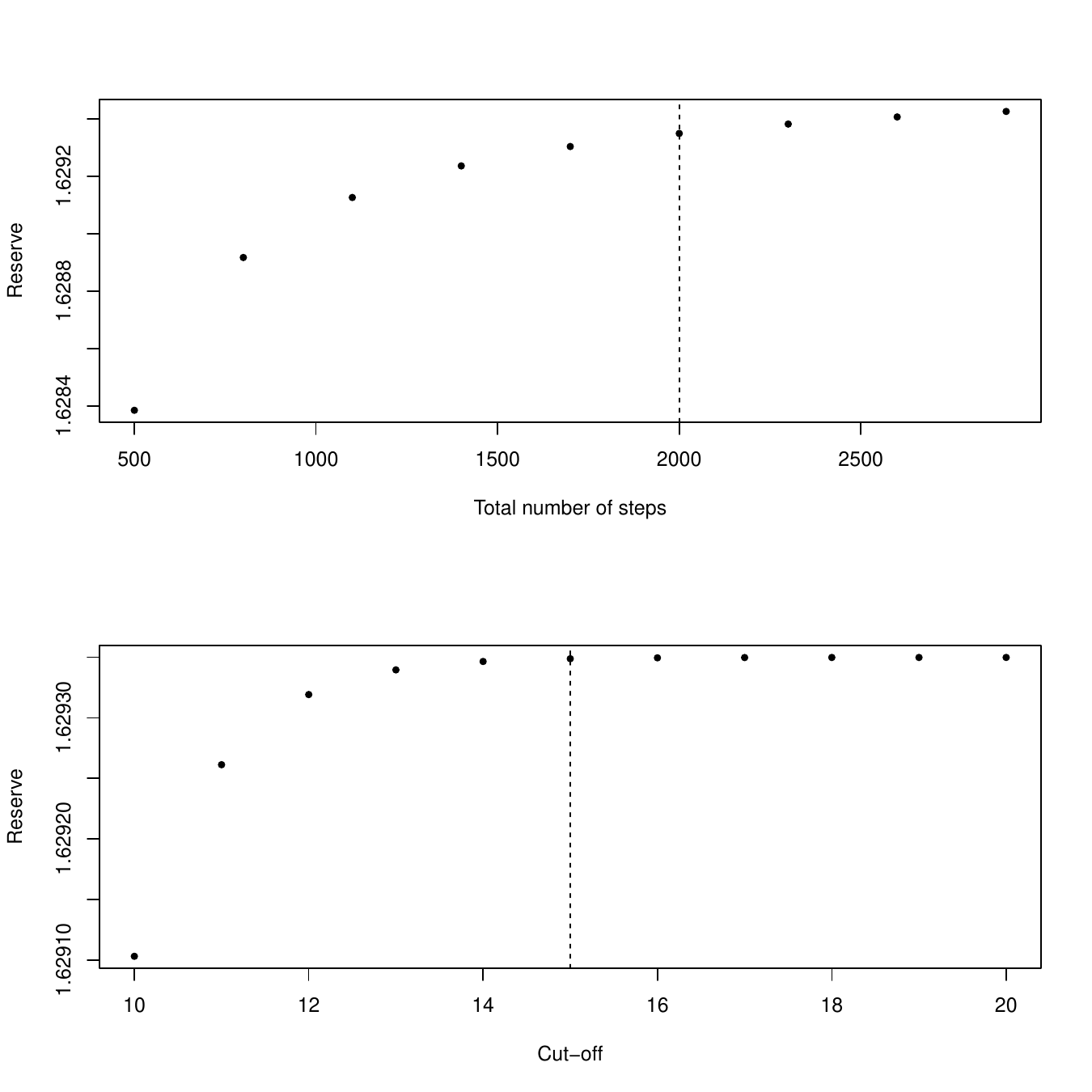}
    \caption{{Convergence of the numerical method as a function of the total number of steps with the fixed cut-off $K_H = 20$ (top) and as a function of the cut-off with $2,000$ total number of steps (bottom); the dotted lines indicate the values selected for the subsequent analyses. \\[4mm]}} 
    \label{fig:conv}
\end{figure}

\begin{table}[h!]
{
\begin{tabular}{r|rrr}   
\hline 
$n$ & Mean-field & Monte Carlo & True \\ \hline
1 & 1.629 & 1.647 & 1.668 \\
2 & 1.629 & 1.651  & --- \\
5 & 1.629 & 1.661  & ---\\
25 & 1.629 & 1.633 & ---\\
50 & 1.629 & 1.630  & ---\\
100 & 1.629 & 1.629 & ---\\ \hline
\end{tabular}}
\caption{Reserves $V^{1,n}$ computed using the mean-field approximation and a naïve Monte Carlo method (with sample size $M = 40,000$), respectively, for $n=1,2,5,25,50,100$.}
\label{tbl:main}
\end{table}

{To further compare the performance of the mean-field approximation to the naïve Monte Carlo method, we include in Table~\ref{tbl:runtime} runtimes corresponding to the computations in Table~\ref{tbl:main}. The runtime measurements are derived from non-optimized code executed on a hardware configuration with an Intel(R) Core(TM) i7-8700 CPU (using one of six cores) as well as 2x16 GB DDR4 2400 MHz RAM. Proper code optimazation should lead to significant improvements for both the mean-field approximation and the Monte Carlo method, but perhaps to a lesser degree for the latter due to its inherent simplicity. In terms of runtimes, we see that the Monte Carlo method (using the large sample size of $M = 40,000$) remains competitive for $n = 5$ but already for $n = 25$ is much slower than the mean-field approximation.}

\begin{table}[t!]
{
\begin{tabular}{r|rrr}  
\hline
$n$ & Mean-field & Monte Carlo \\ \hline
1 & 1m34s & 30s  \\
2 & 1m34s & 45s \\
5 & 1m34s & 1m49s \\
25 & 1m34s & 10m10s \\
50 & 1m34s & 22m58s  \\
100 & 1m34s & 52m30s \\ \hline
\end{tabular}}
\caption{{Runtimes using the mean-field approximation and a naïve Monte Carlo method (with sample size $M = 40,000$), respectively, for $n=1,2,5,25,50,100$.}}
\label{tbl:runtime}
\end{table}

The deviations between the mean-field approximation and the Monte Carlo estimates {for moderate $n$} could indicate that the mean-field approximation is somewhat poor, that the Monte Carlo estimate has yet to fully converge, or both. Table~\ref{tbl:mc} contains statistics for $50$ repeated applications of the naïve Monte Carlo method with $n=2,5,25$. The quantiles and the standard deviations both indicate substantial variance in the Monte Carlo estimates, and for $n=5,25$ the mean-field approximation is contained in the empirical $90\,\%$ confidence intervals. {This confirms that $M = 40,000$ does not suffice to ensure the convergence of the Monte Carlo estimates. Therefore, we must conclude that the mean-field approximation constitutes a necessary, and rather efficient, alternative already for moderate $n$.}

\begin{table}[h!]
{
\begin{tabular}{r|rrrr}
\hline
$n$ & Second lowest & Average & Second highest & Standard deviation \\ \hline
2 & 1.632 & 1.651 &1.673 & 0.0122 \\
5 & 1.627 & 1.643 & 1.661 & 0.0087 \\
25 & 1.626 & 1.633 & 1.639 & 0.0035 \\ \hline
\end{tabular}}
\caption{Statistics for $50$ repeated applications of the naïve Monte Carlo method to estimate $V^{1,n}$ with $n=2,5,25$.}
\label{tbl:mc}
\end{table}

\subsection{Further findings} \label{sec:sim_findings}

The reserve in the one-individual model is about $2.38\,\%$ larger than the mean-field reserve. This is because in the one-individual model, the effect of health claims on the disability rate is calculated based only on the health claims history of a single individual, meaning that the (likely) occurrence of just one health claim causes the disability rate to spike quite violently upwards, confer also with Figure~\ref{fig:dis_cred_one}. Since this is not similarly counteracted by downwards movements, there is a larger probability of disability in the one-individual model compared to the mean-field approximation and consequently also a larger reserve.

The mean-field convergence implies that $\nu^n \to v$ as $n \to \infty$. This convergence is neatly illustrated in Figure~\ref{fig:dis_cred_conv}, which mirrors Figure~\ref{fig:dis_cred_one}, but has $n=5,25,100$ rather than $n = 1$. Substantial deviations in disability rate appear even for $n=25$.

\pagebreak

\begin{figure}[h!]
    \centering
    \includegraphics[width=0.75\textwidth]{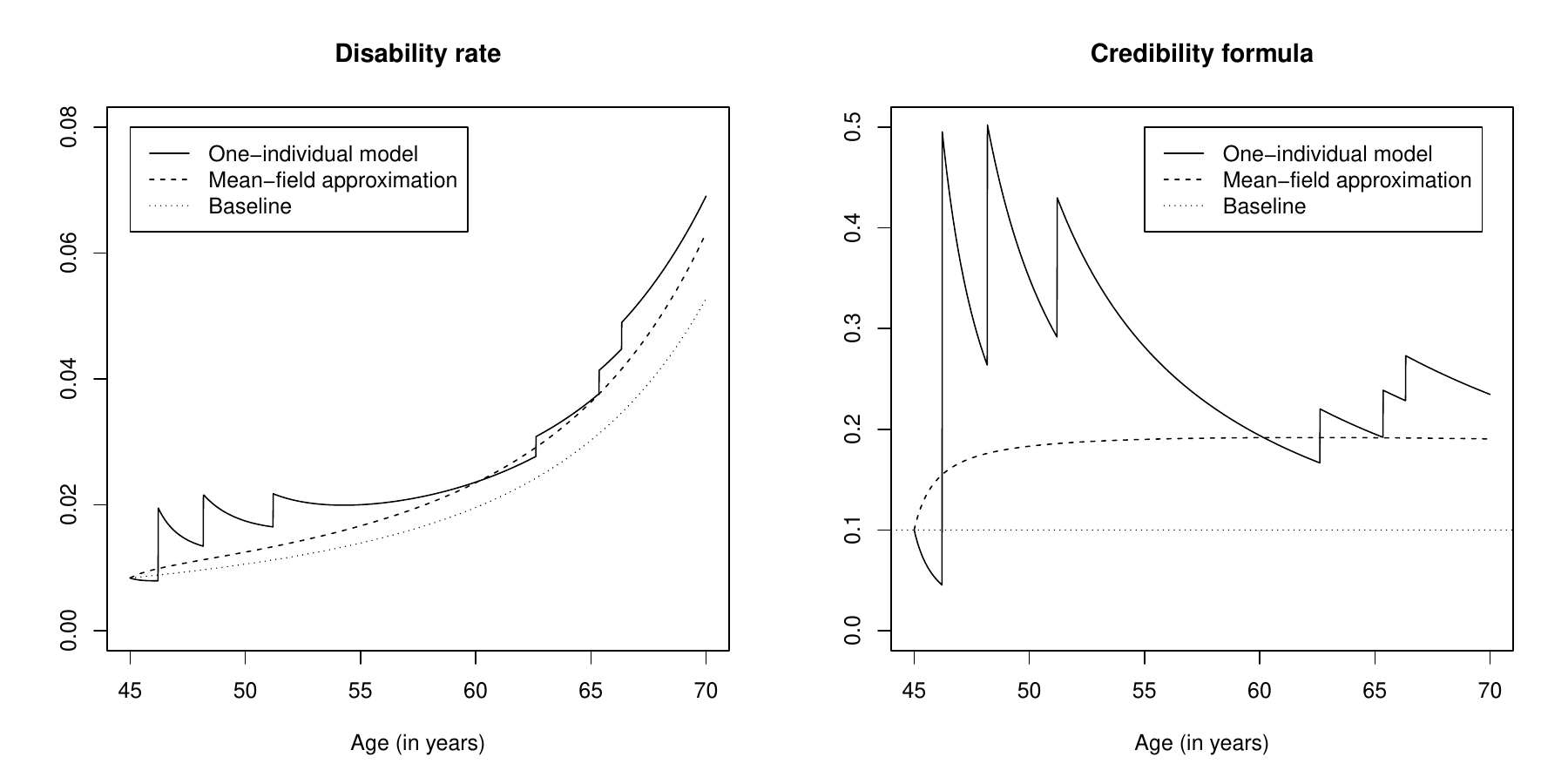}
    \caption{Disability rate (left) and credibility formula (right) for a single realization of the one-individual model ($y = \nu^1$) with the mean-field approximation ($y = v$) and the baseline ($y = \zeta_1$).} 
    \label{fig:dis_cred_one}
\end{figure}

\begin{figure}[h!]
    \centering
    \includegraphics[width=0.75\textwidth]{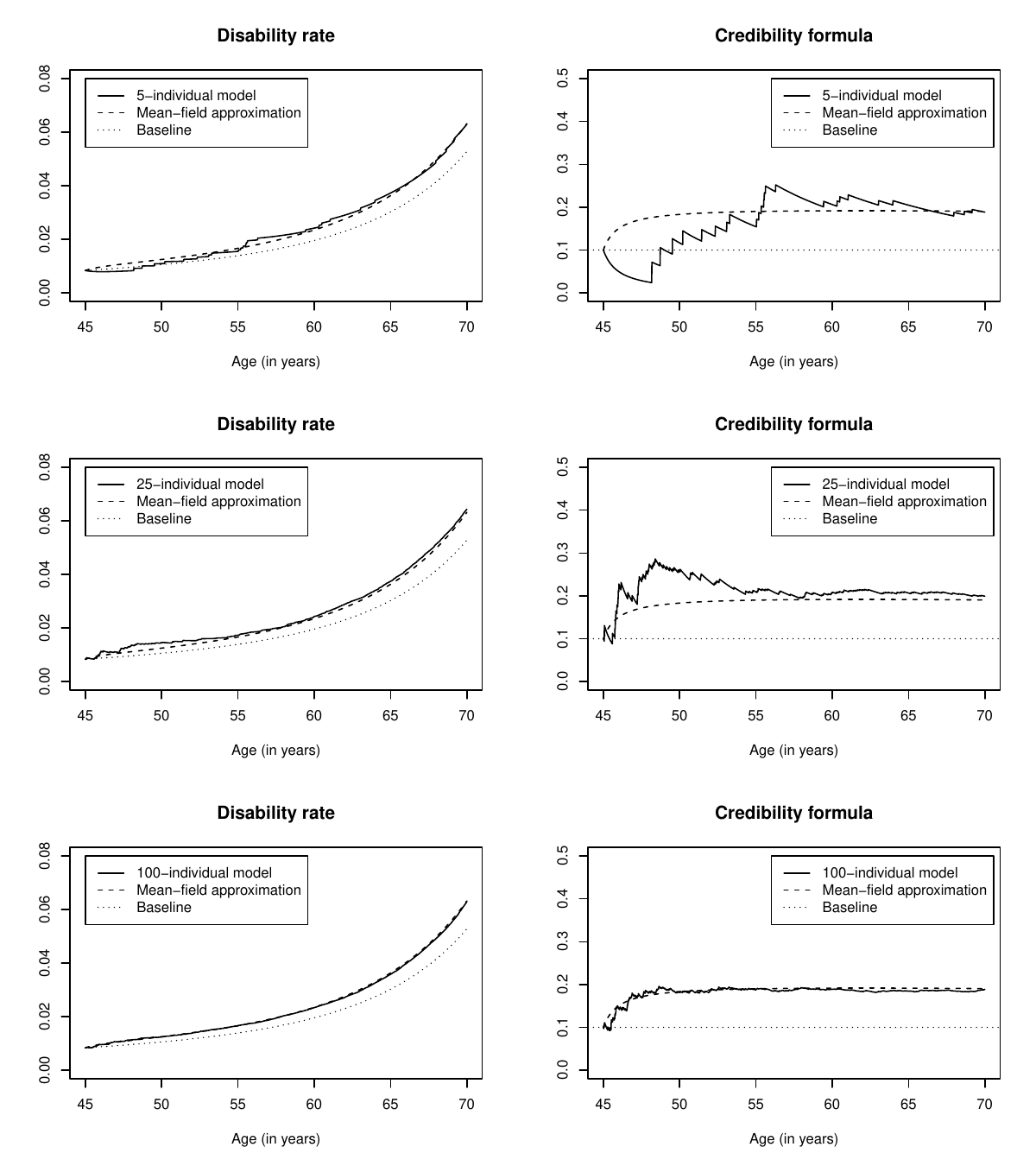}
    \caption{For $n=5,25,100$ disability rates (left) and credibility formulas (right) for a single realization of the $n$-individual model ($y = \nu^n$) with the mean-field approximation ($y = v$) and the baseline ($y = \zeta_1$).} 
    \label{fig:dis_cred_conv}
\end{figure}

\pagebreak

To assess the quality of the convergence, we may additionally study histograms of the average present values
\begin{align*}
\frac{1}{n}\sum_{\ell=1}^n \int_0^T e^{-\int_0^t r(s) \, \mathrm{d}s} B^{\ell,n,m}(\mathrm{d}t)
\end{align*}
for $m=1,\ldots,M$, where $M = 40,000$. The resulting histograms for $n=5,25,50,100$ can be found in Figure~\ref{fig:histograms}. For $n \geq 50$ the clear shape of a bell curve emerges, which would seem to indicate that besides the laws of large numbers already covered in Section~\ref{sec:mean_field}, a central limit theorem might also hold. {Proposition~6.7} in~\cite{Hornung2025} confirms exactly such a central limit theorem, but subject to a covariance condition that is difficult to verify theoretically. Figure~\ref{fig:histograms} offers empirical support for the conjecture that, in this specific model, the condition is met.

\begin{figure}[h!]
    \centering
    \includegraphics[width=0.95\textwidth]{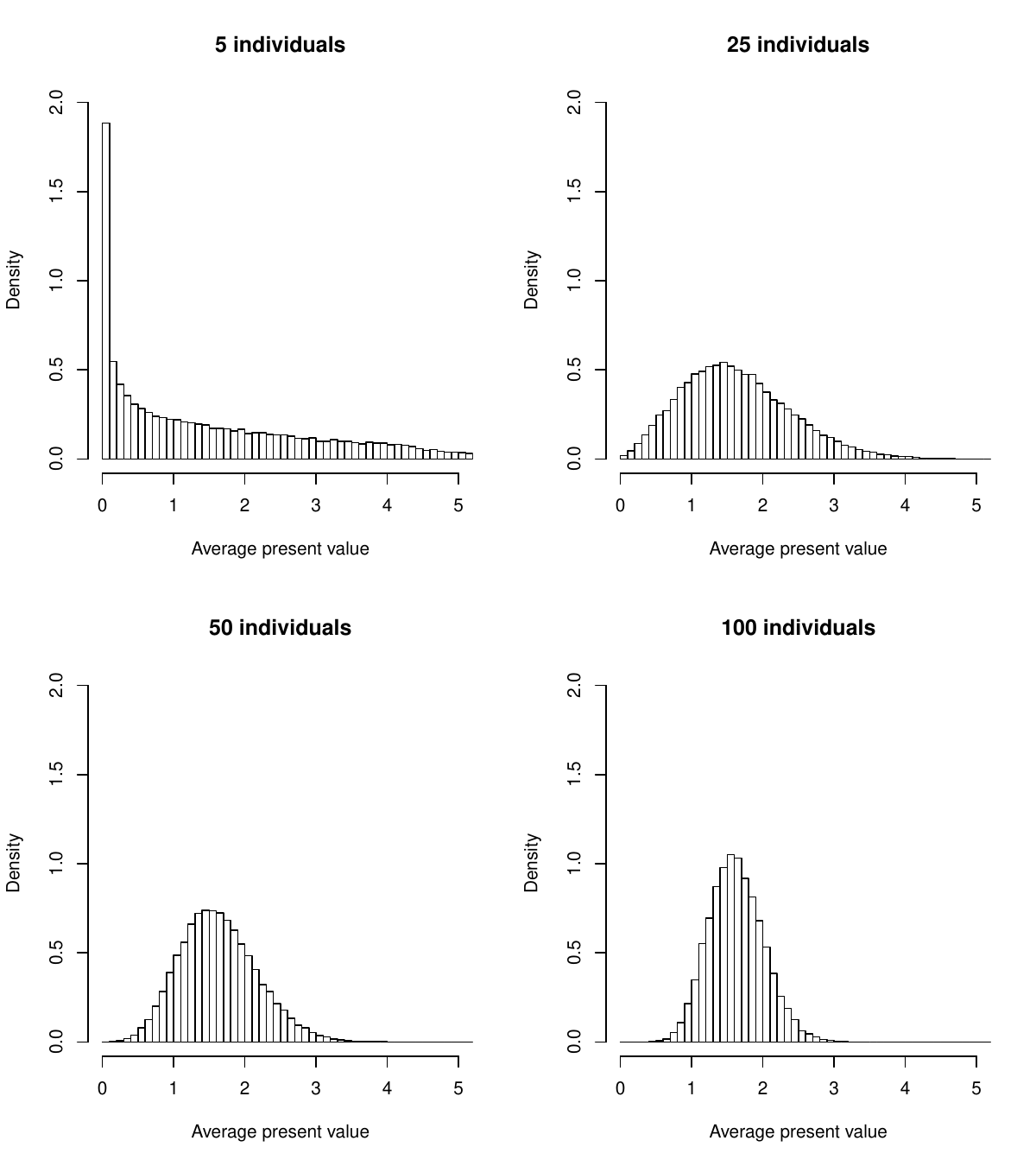}
    \caption{Histograms of average present values for $n=5,25,50,100$; as $n$ increases, a clear bell curve emerges.} 
    \label{fig:histograms}
\end{figure}

{

\subsection{Inclusion of lapse risk}\label{sec:lapse}

Throughout, we have made the assumption that the group of insured is closed -- that is, members do not leave and new members do not join. This is not realistic for the application we have in mind, where group members constitute employees and the group a specific employer. In this subsection, we study the impact of lapse risk on the calculation of the reserve; besides lapse risk affecting the individual, lapses may change the group composition and hereby the collective effect in the disability rate.

We suppose that lapses occur solely from state $1$, that is when the individual is active, and denote the rate of lapse by $t \mapsto \alpha(t)$. We consider three different ways of including lapse risk:
\begin{itemize}[leftmargin=*]
\item[] \textbf{Approach~A -- individually.} Lapse risk affects only the individual and not the collective: We first calculate $\bar{p}$ and $v$ excluding lapse risk, that is with $\mu_{1\bigcdot}$ unchanged, and then we repeat the calculation of $\bar{p}$ using this $v$ and including lapse risk, that is adding $\alpha$ to $\mu_{1\bigcdot}$ \\[-2mm]
\item[] \textbf{Approach~B -- collectively.} Lapse risk affects both the individual and the collective: We calculate $\bar{p}$ and $v$ including lapse risk, that is adding $\alpha$ to $\mu_{1\bigcdot}$ \\[-2mm]
\item[] \textbf{Approach~C -- collectively with adjustment.} Lapse risk affects both the individual and the collective, but with an adjustment to account for the effect it has on the group composition: We calculate $\bar{p}$ and $v$ including lapse risk, that is adding $\alpha$ to $\mu_{1\bigcdot}$, but we adjust the formula of $v$ according to
\begin{align*}
v(t) = \frac{\sum_j \sum_{h=0}^\infty \max\{h,\kappa_H\}\bar{p}_{1j}(t,t,h)}{\sum_j \sum_{h=0}^\infty \bar{p}_{1j}(t,t,h)}.
\end{align*}  
Note that -- since $\bar{p}$ has been calculated with $\alpha$ added solely to the diagonal element $\mu_{1\bigcdot}$ -- the denominator is less than one.
\end{itemize}
Recall that $v$ is, practically speaking, describing the expected number of health claims of a typical individual in the group. Introducing lapse risk, one might ask: What do we understand by `group' and does its composition change over time? If we do not adjust the formula of $v$, the one understands the group as composed of both lapsed individuals (with future health claims unrecorded) and individuals yet to lapse. If we adjust the formula of $v$ as in Approach~C above, one understands the group as solely composed of the individuals yet to lapse.

In Table~\ref{tbl:lapse}, we present mean-field approximations of the reserve for constant lapse rates of $0.02$, $0.05$, and $0.1$. We see, as one would expect, that the reserve decreases as the lapse rate increases. The inclusion of lapse risk produces similar results in Approach~A and~C, while Approach~B admits even lower reserves. This is because increasing the lapse rate not only decreases the likelihood that the individual utilizes their disability insurance, but also decreases the collective effect on the disability rate through the unadjusted $v$.

\begin{table}[h!]
{
\begin{tabular}{r|rrr}  
\hline
$\alpha$ & Individually & Collectively & Collectively with adjustment \\ \hline
0 \% & 1.629 & 1.629 & 1.629 \\
2 \% & 1.311 & 1.271 & 1.310 \\
5 \% & 0.974 & 0.916 & 0.973 \\
10 \% & 0.639 & 0.586 & 0.637\\ \hline
\end{tabular}}
\caption{{Reserves $V^{1,n}$ computed using the mean-field approximation and including lapse risk across Approach~A,~B, and~C.}}
\label{tbl:lapse}
\end{table}

To further illuminate the different approaches, we plot in Figure~\ref{fig:lapse} the disability rates and credibility formulas for a lapse rate of $0.05$. The observed effects are consistent with the results for the reserve; however, we see that although the reserves in Approach~A and~C are virtually indistinguishable, the underlying credibility formulas are indeed different. This is because the implicitly assumed group compositions are different, with there being a higher proportion of disabled and dead individuals in Approach~C compared to Approach~A; this is due to the fact that only active individuals may lapse. 

The most important conclusion to draw from the findings presented in this subsection is the following. There are indeed discrepancies in the reserves calculated across Approach~A,~B, and~C, and these discrepancies stem from differences in implicit assumptions regarding the composition of the group. Which approach is appropriate depends on how changes to the composition of the group have been taken into account in the estimation of the disability intensity. This is about properly adjusting $\nu^n$ in the presence of samling effects so it remains observable yet still related to the estimand of interest.

\begin{figure}[h!]
    \centering
    \includegraphics[width=0.95\textwidth]{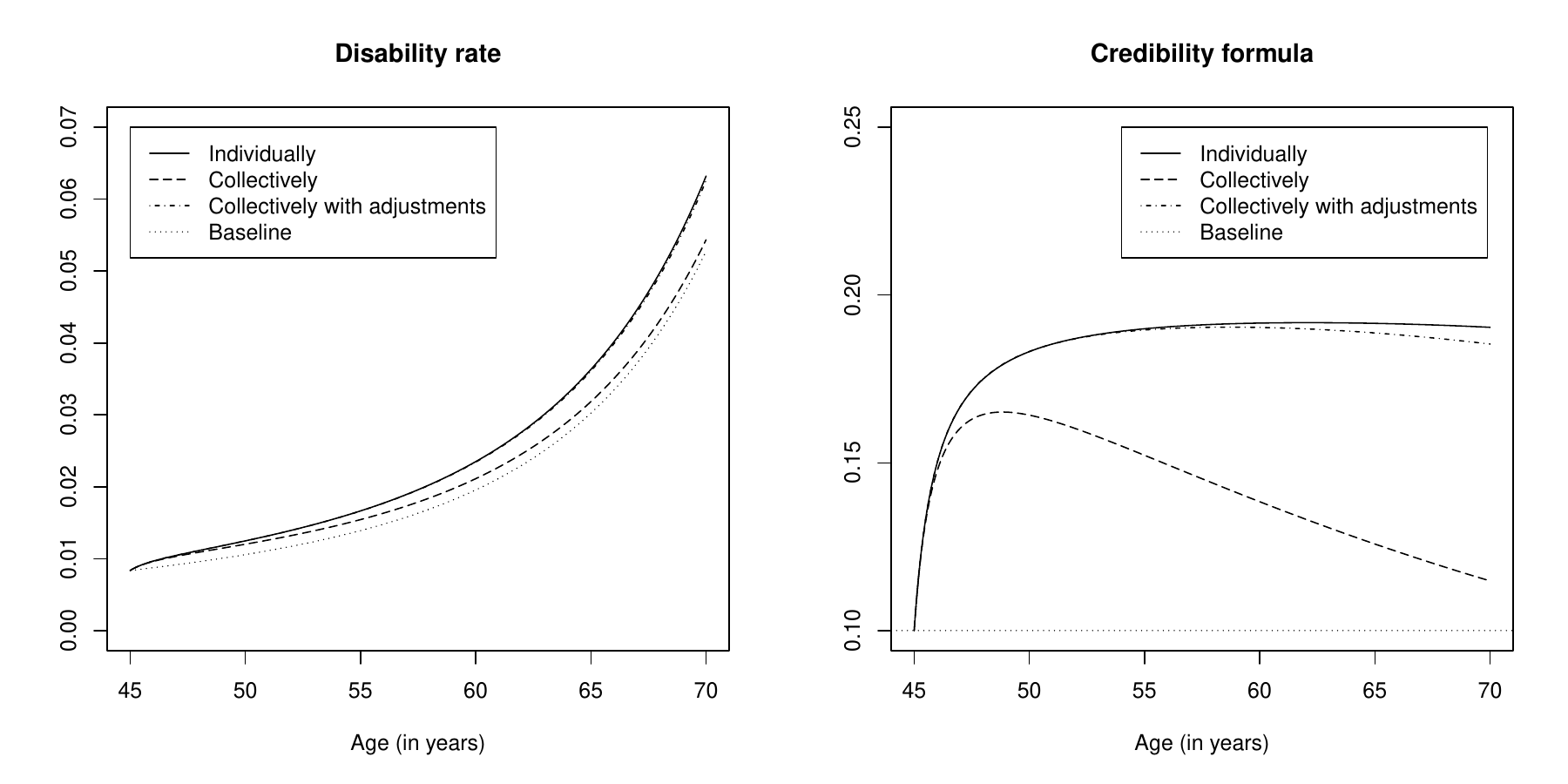}
    \caption{{Disability rate (left) and credibility formula (right) for a lapse rate of $0.05$ and across Approach~A,~B, and~C together with the baseline.}} 
    \label{fig:lapse}
\end{figure}

}

\section{Outlook}\label{sec:conclude}

{
In Subsection~\ref{sec:lapse}, we gave a practical perspective on how to account for lapse risk. The idea was to make plausible that the methods developed in this paper for closed groups may also be adapted to the situation where members leave and new members join. From a more theoretical perspective, to capture such policyholder behavior and sampling effects, one can expand the state-space to encompass group entries and exits. In principle, this is not at odds with the mean-field theory or the likelihood-based estimation, except that one -- as illustrated in Subsection~\ref{sec:lapse} -- may have to adjust $\nu^n$ and $v$ to the situation at hand. The detailed mathematical and statistical analysis of such adjustments, including alternative choices of $\nu^n$, constitute an interesting avenue for future work.}

In this paper, we focus on valuation at contract inception by firmly placing ourselves at initial time $t_0 = 0$ in the calculation of reserves, and we make the simplifying assumption that $U^{\ell,n}_0 = H^{\ell,n}_0 = 0$ and thus $\bar{U}_0 = \bar{H}_0 = 0$. To establish mean-field convergence, we only needed the assumption of $\pi$-chaosticity for $(Z^{1,n}_0,\ldots,Z^{n,n}_0)(\amsmathbb{P})$ for some $\pi$. If we instead place ourselves at initial time $t_0>0$, the values $U^{\ell,n}_{t_0}$, $H^{\ell,n}_{t_0}$, $\bar{U}_{t_0}$, and $\bar{H}_{t_0}$ are random. {If everything develops according to the model, one might like to initialize with the distribution of $(X^{1,n}_{t_0},\ldots,X^{n,n}_{t_0})$. It can be shown from Theorem~\ref{thm:mean-field_convergence} and properties of the total variance distance that the distribution of $(X^{1,n}_{t_0},\ldots,X^{n,n}_{t_0})$ is $\bar{X}_{t_0}(\amsmathbb{P})$-chaotic. Consequently, the distribution of $(X^{1,n}_{t_0},\ldots,X^{n,n}_{t_0})$ constitues a theoretically sound choice of initial distribution, and most of the results for $t_0 = 0$ thus also hold for $t_0 > 0$. However, it should be noted that it may be technically problematic to condition on null-sets -- such as a specific duration -- which poses a challenge for the state-wise quantities, confer also with the discussion between Corollary~5.5 and Theorem~5.6 in~\cite{Hornung2025}.  In practice, however, this can be resolved by instead conditioning on a sufficiently small duration interval.}

\section*{Acknowledgments}

Both authors have carried out this research in association with the project frame InterAct.

\section*{Disclosure statement}

The expressed opinions are attributable solely to us and do not necessarily reflect the views of any of our past, current, and future employers.

\bibliographystyle{amsplain}
\bibliography{references.bib}

@article{Aalen1987,
  author = {O.O. Aalen},
  journal = {Scandinavian Actuarial Journal},
  pages = {177-190},
  title = {{Dynamic modelling and causality}},
  doi = {10.1080/03461238.1987.10413826},
  volume = {1987},
  number = {3-4},
  year = {1987}}

@book{AndersenBorganGillKeiding1993,
  Author    = {P.K. Andersen and {\O}. Borgan and R.D. Gill and N. Keiding},
  Title     = {{Statistical Models Based on Counting Processes}},
  Publisher = {Springer, New York},
  Series   = {Springer Series in Statistics},
  doi	= {10.1007/978-1-4612-4348-9},
  Year      = {1993}}

@article{BuchardtMollerSchmidt2015,
  doi = {10.1080/03461238.2013.879919},
  journal = {Scandinavian Actuarial Journal},
  volume = {2015},
  number = {8},
  year = {2015},
  pages = {660-688},
  title = {{Cash flows and policyholder behaviour in the semi-Markov life insurance setup}},
  author = {K. Buchardt and T. M{\o}ller and K.B. Schmidt}}

@article{Christiansen2012,
author = {M.C. Christiansen},
title = {Multistate models in health insurance},
journal = {Advances in Statistical Analysis},
volume = {96},
pages = {155-186},
year  = {2012},
doi = {10.1007/s10182-012-0189-2}}

@article{DjehicheLoefdahl2016,
  author = {B. Djehiche and B. L{\"o}fdahl},
  title = {{Nonlinear reserving in life insurance: Aggregation and mean-field approximation}},
  journal = {Insurance: Mathematics and Economics},
  year = {2016},
  volume = {69},
  pages = {1-13},
  doi = {10.1016/j.insmatheco.2016.04.002}}

@article{FeinbergMandavaShiryaev2014,
  doi = {10.1016/j.jmaa.2013.09.043},
  journal = {Journal of Mathematical Analysis and Applications},
  volume = {411},
  year = {2014},
  pages = {261-270},
  title = {{On solutions of Kolmogorov's equations for nonhomogeneous jump Markov processes}},
  author = {E.A. Feinberg and M. Mandava and A.N. Shiryaev}}

@article{FeinbergMandavaShiryaev2022,
  doi = {10.1007/s10479-017-2538-8},
  journal = {Annals of Operations Research},
  volume = {317},
  year = {2022},
  pages = {587-604},
  title = {{Kolmogorov’s equations for jump Markov processes with unbounded jump rates}},
  author = {E. Feinberg and M. Mandava and A.N. Shiryaev}}

@article{Furrer2019,
	title={{Experience rating in the classic Markov chain life insurance setting: An empirical Bayes and multivariate frailty approach}},
	author={Furrer, C.},
	journal={European Actuarial Journal},
	volume={9},
	pages={31-58},
	doi={10.1007/s13385-019-00190-5},
	year={2019}}

@article{FurrerSoerensenYslas2025,
	title={{Bivariate phase-type distributions for experience rating in disability insurance}},
	author={C. Furrer and J.J. S{\o}rensen and J. Yslas},
	journal={European Actuarial Journal},
	pages={1-37},
	doi={10.1007/s13385-025-00439-2},
	year={2025},
	note={online first}}

@phdthesis{Helwich2008,
  author       = {M. Helwich},
  title        = {Durational effects and non-smooth semi-{M}arkov models in life insurance},
  school       = {University of Rostock},
  year         = {2008}}

@incollection{Hoem1972,
title = "Inhomogeneous {S}emi-{M}arkov {P}rocesses, {S}elect {A}ctuarial {T}ables, and {D}uration-{D}ependence in {D}emography",
editor = "T.N.E. Greville",
booktitle = "Population Dynamics",
publisher = "Academic Press",
pages = "251-296",
year = "1972",
doi = "10.1016/B978-1-4832-2868-6.50013-8",
author = "J.M. Hoem"}

@misc{Hornung2025,
  title={{Mean-field approximations in insurance}},
  author={Hornung, P.C.},
  note={arXiv:2511.04198},
  year={2025}}

@book{Jacobsen2006,
	author    = {M. Jacobsen},
	title     = {{Point process theory and applications: Marked point and piecewise deterministic processes}},
	publisher = {Birkh{\"a}user, Boston},
	series   = {Probability and its Applications},
	doi	= {10.1007/0-8176-4463-6},
	year      = {2006}}

@article{LewisShedler1979,
	title={{Simulation of nonhomogeneous Poisson processes by thinning}},
	author={Lewis, P.A.W. and Shedler, G.S.},
	journal={Naval Research Logistics Quarterly},
	volume={26},
	number={3},
	pages={403-413},
	year={1979}}

\end{document}